\documentclass{article}

\title{Pricing energy spread options with variance gamma-driven Ornstein-Uhlenbeck dynamics}
\author{Tim Leung\thanks{Department of Applied Mathematics, University of Washington, Seattle, WA 98195, USA. Email: \href{mailto:timleung@uw.edu}{timleung@uw.edu}.}
    \and Kevin W.\ Lu\thanks{Research School of Finance, Actuarial Studies \& Statistics, Australian National University, Canberra,  ACT 0200, Australia.  Email: \href{mailto:kevin.lu@anu.edu.au}{kevin.lu@anu.edu.au}.}}
\date{\today}

\usepackage{natbib}
\usepackage{graphicx}
\usepackage{amsmath,amsfonts,amssymb,amsthm}
\usepackage[colorlinks,citecolor=red,linkcolor=blue]{hyperref}
\usepackage{float}
\usepackage{mathrsfs}
\usepackage{siunitx}
\usepackage{xcolor}
\numberwithin{equation}{section} 
\allowdisplaybreaks[4] 
\let\originalleft\left
\let\originalright\right
\def\left#1{\mathopen{}\originalleft#1}
\def\right#1{\originalright#1\mathclose{}}
\definecolor{darkgreen}{rgb}{0.00, 0.55, 0.00}

\newtheorem{theorem}{Theorem}[section]
\newtheorem{corollary}[theorem]{Corollary}
\newtheorem{lemma}[theorem]{Lemma}
\newtheorem{proposition}[theorem]{Proposition}

\theoremstyle{definition}

\theoremstyle{remark}
\newtheorem{remark}[theorem]{Remark}

\newcommand{\rmd}{{\rm d}}
\newcommand{\rmi}{{\rm i}}

\newcommand{\eqd}{\stackrel{D}{=}}
\newcommand{\halmos}{\quad\hfill\mbox{$\Box$}}

\newcommand{\CC}{\mathbb{C}}
\newcommand{\DD}{\mathbb{D}}
\newcommand{\EE}{\mathbb{E}}
\newcommand{\FF}{\mathbb{F}}

\newcommand{\RR}{\mathbb{R}}

\newcommand{\PP}{\mathbb{P}}
\newcommand{\QQ}{\mathbb{Q}}

\newcommand{\BBB}{{\cal B}}

\newcommand{\DDD}{{\cal D}}

\newcommand{\FFF}{{\cal F}}

\newcommand{\ZZZ}{{\cal Z}}

\newcommand{\skal}[2]{\langle #1,#2\rangle}

\newcommand{\eins}{{\bf 1}}

\newcommand{\bfnull}{{\bf 0}}

\newcommand{\bfb}{{\bf b}}
\newcommand{\bfc}{{\bf c}}

\newcommand{\bfe}{{\bf e}}
\newcommand{\bfh}{{\bf h}}
\newcommand{\bfi}{{\bf i}}
\newcommand{\bfj}{{\bf j}}

\newcommand{\bft}{{\bf t}}
\newcommand{\bfx}{{\bf x}}
\newcommand{\bfy}{{\bf y}}

\newcommand{\bfz}{{\bf z}}

\newcommand{\bfB}{{\bf B}}

\newcommand{\bfG}{{\bf G}}
\newcommand{\bfH}{{\bf H}}

\newcommand{\bfR}{{\bf R}}
\newcommand{\bfS}{{\bf S}}
\newcommand{\bfT}{{\bf T}}

\newcommand{\bfV}{{\bf V}}

\newcommand{\bfX}{{\bf X}}
\newcommand{\bfY}{{\bf Y}}
\newcommand{\bfZ}{{\bf Z}}
\newcommand{\bfalpha}{\boldsymbol{\alpha}}

\newcommand{\bfepsilon}{\boldsymbol{\epsilon}}

\newcommand{\bfLambda}{\boldsymbol{\Lambda}}

\newcommand{\bfmu}{\boldsymbol{\mu}}

\newcommand{\bfSigma}{{\boldsymbol{\Sigma}}}

\newcommand{\bftheta}{\boldsymbol{\theta}}
\newcommand{\bfvtheta}{\boldsymbol{\vartheta}}
\newcommand{\bfeta}{\boldsymbol{\eta}}

\newcommand{\myCov}{\operatorname{Cov}}
\newcommand{\myVar}{\operatorname{Var}}
\newcommand{\myCorr}{\operatorname{Corr}}

\newcommand{\wt}{\widetilde}
\newcommand{\tr}{\diamond}

\newcommand{\given}{{\,\vert\,}}

\newcommand{\wh}{\widehat}

\newcommand{\name}{\operatorname}

\DeclareMathOperator*{\argmin}{arg\,min}

\begin{document}

    \maketitle
    
    \begin{abstract}

        We consider the pricing of energy spread options for spot prices following an exponential Ornstein-Uhlenbeck process driven by a sum of independent multivariate variance gamma processes, which gives rise to mean-reverting, infinite activity price dynamics.   Within this class of driving processes, the Esscher transform is used to obtain an equivalent martingale measure with a focus on the weak variance alpha-gamma process. By deriving an analytical formula for the cumulant generating function of the innovation term, we obtain a pricing formula for forwards and apply the FFT method of Hurd and Zhou to price  spread options. Lastly, we demonstrate how the model should be both estimated on energy prices under the real world measure and calibrated on forward or call prices, and provide numerical results for the pricing of spread options.

    \end{abstract}

    \section{Introduction}
    Energy is a highly volatile commodity. Electricity markets in particular experience large price fluctuations due to temporary imbalances in supply and demand, weather events, unplanned power plant outages, and the intermittency of renewable energy generation.  Retailers are exposed to this market risk as they buy electricity at the spot price and sell to customers at a fixed price. Energy derivatives are used to manage these risks, and this is the main topic of our study.   A detailed monograph is provided by Benth, {\v S}altyt{\.e} Benth and Koekebakker \cite{benthbook}.

    Specifically, we focus on the pricing of spread options, which are  European-style derivatives on a bivariate price process $(S_1,S_2)$ whose payoff is
    \begin{align*}
        f_T = (S_1(T)-S_2(T)-K)^+
    \end{align*}
    with strike price $K$ and maturity time $T$.   Energy prices tend to be highly correlated across commodities and regions. For example, electricity prices are related to the price of gas used to generate it. An option on this difference is called a spark spread, which are used by generators to hedge against unanticipated spikes in production costs. As another example, the physical interconnectors between the five states of the National Electricity Market (NEM) in Australia also induce a correlation, and an option on such a difference would be a regional spread. The survey by Carmona and Durrleman \cite{cardur03} catalogs the variety of spread options. In general, there is no known closed-form pricing formula for spread options, except when  $K=0$ in a lognormal model where it reduces to the well-known Margrabe formula.

    To approach this problem, we begin with the stylized facts that energy prices exhibit seasonality, mean reversion, and jumps. To account for these features, we assume that the $n$-dimensional energy spot price process $\bfS = (\bfS(t))_{t\in[0,T]}$  under the real world measure $\PP$ follows
    \begin{align}
        \bfS(t) = \exp(\boldsymbol{\Lambda}(t) +\bfX(t)), \quad t\in[0,T],\label{priceproc}
    \end{align}
    where $\bfLambda(t)$ is a deterministic seasonality function, and $\bfX$ is a L\'evy-driven Ornstein-Uhlenbeck process (LDOUP) defined by
    \begin{align}
        \rmd\bfX(t) =-\lambda\bfX(t)\rmd t + \rmd \bfZ(\lambda t),\quad \bfX(0)=\bfX_0, \label{ldoup}
    \end{align}
    where $\bfZ$ is the \emph{background driving L\'evy process (BDLP)}, $\lambda>0$ is  the speed of mean reversion, and the independent random vector $\bfX_0$ is the initial value. The  Ornstein-Uhlenbeck  part ensures mean reversion, while the L\'evy part allows for jumps, and  indeed the use of LDOUPs is widespread in energy price modeling. For example, see \cite{BKMB, benthbook,BenthPir18,cf05,sab2020,sab2023}.
    
    The BDLP $\bfZ$ we use herein is a sum of independent multivariate  variance gamma (VG) processes with a focus on the weak variance alpha-gamma (WVAG) process by Buchmann, Lu and Madan \cite{BLM17a} as a subclass. A multivariate VG process is a multivariate Brownian motion $\bfB$ subordinated, or time-changed, by a univariate gamma subordinator $G$, and such processes originated with Madan and Seneta  \cite{MaSe90} and have been studied in \cite{BKMS16,FS07}. While a sum of VG processes does not necessarily preserve such a time-change interpretation, the WVAG process being originally defined via weak subordination does. It turns out to be equivalent to a process $\bfZ$ characterized in the following way: consider a subordinator $\bfT$ which is a sum of $G_k$ representing the idiosyncratic  time-change for the $k$th price, and $\bfalpha G_0$ representing a common time change for all prices, where $G_0,\dots,G_n$ are independent gamma subordinators and $\bfalpha\in[0,\infty)^n$, then when  $\bfT(t)$ jumps by $\Delta \bfT(t)$,  $\bfZ(t)$  jumps by the componentwise composition $\bfB(\Delta \bfT(t))$ in distribution, rather than pathwise for traditional subordination. Consequently, the WVAG process, using weak subordination, preserves the essence of traditional subordination in a distributional or weak sense.
    
    This provides a naturally multivariate process with a flexible dependence structure where the Brownian motion has arbitrary covariance in contrast to other approaches. Specifically, these include Luciano and Semeraro \cite{LS10} who in one model use traditional subordination restricted to independent Brownian motions to remain a L\'evy process, Ballotta and Bonfiglioli \cite{BB14} who use summing independent univariate VG processes $V_0,\dots,V_n$ in the form $\bfZ = (\alpha_1V_0+V_1,\dots, \alpha_n V_0 + V_n)$, an approach that is also found in \cite{CaFu2013,HZ10}, and  Gardini, Sabino and Sasso \cite{gss2021} who use the above two approaches with dependent subordinators constructed by self-decomposability.  The last three references include applications to spread option pricing. Our BLDP $\bfZ$, being a sum of VG processes, includes the cases of \cite{LS10,BB14} but not \cite{gss2021}. A similar way of summing univariate LDOUPs to produce a bivariate energy price for spread option pricing is also considered in \cite[Section 9.2.1]{benthbook}. Our LDOUP $\bfX$ include certain restrictions of this model, although both models are more general than each other in different ways.
    
    So in all, our model of the multivariate energy spot prices allows for seasonality, mean reversion, has infinite activity due to the VG-related BDLP, and has a naturally multivariate time-change interpretation.

        Energy prices such as electricity are considered nontraded assets in standard arbitrage theory since electricity bought or sold on the market is instantly consumed and cannot be dynamically and frictionlessly traded to construct a replicating portfolio.
        Since the market is incomplete, there are infinitely many EMMs consistent with NA, so we need to postulate a model for it.  Accordingly, for the pricing of energy derivatives, we use the Esscher transform \cite{gs94} to specify a parametric family of EMMs, known as the {Esscher measure}  $\QQ_{\bfh}$  indexed by the market price of risk (MPR) $\bfh$. 
        This is a common approach for L\'evy-based models in energy markets \cite{benthbook}, as well as equity markets \cite{pasc} in which the incompleteness is due to the presence of jumps.  A further difference between energy and equity markets is that the discounted energy price process is not required to be a  $\QQ_{\bfh}$-martingale, a condition in an equity market that  would ordinarily determine a unique $\bfh$, and hence an EMM $\QQ_{\bfh}$.  Instead, $\bfh$ must be calibrated to, for example, observed forward or call prices as discussed in Section \ref{imple} below.

    Here, we show how the BDLP $\bfZ$ changes under the Esscher transform into a sum of VG processes with modified parameters, which leads to the result that the class of WVAG processes is not closed under Esscher transform. In other related works, Buchmann et al.\ \cite{BKMS16} studies Esscher transforms for a class of strongly subordinated Brownian motions, and  Michaelsen and Szimayer \cite{MiSz17} provides numerical results about pricing options on equities that follow an exponential WVAG process using the Esscher transform, although these rely on results that do not extend to energy derivatives. Our results are more general.
    
    We derive an analytical formula for the cumulant generating function (cgf) of the innovation term of a LDOUP where the BLDP is a multivariate VG process in terms of the dilogarithm function, which generalizes a result in the univariate case by Sabino \cite{sab2020}. This is a key result given in Theorem \ref{innovcgf}, and is  surprising in light of Remark \ref{cgfrem}. Consequently, we obtain an analytical formula for the cgf of the log return along with conditions for this to be finite and when it can be analytically continued into the complex plane. We use this  cgf  to obtain a  forward pricing formula, and combined with the Fourier methods of Carr and Madan \cite{CM99} and Hurd and Zhou \cite{HZ10}, it also gives pricing formulas for  call and spread options, respectively.

   We place our approach to spread option pricing into the context of the existing literature.  In a lognormal model, pricing has been studied using closed form approximations   \cite{cardur03,lzd10},  which started with the well-known Kirk's approximation,  and Fourier methods  \cite{pelsab2014}.     In the general setting using characteristic functions, closed form approximations  have also been developed  \cite{CaFu2013,pellegrino2016}, but contrasts with the Hurd-Zhou method as the latter has arbitrary precision. Our approach to pricing spread options differs from Benth and {\v S}altyt{\.e}-Benth \cite{bsb06} where the  spread is  modeled directly as a univariate LDOUP as opposed to being based on \eqref{priceproc}. It differs from   Gardini and Sabino \cite{gs22} and  Gardini, Sabino and Sasso \cite{gss22b,gss2021} where the spread is  on forwards and the option is priced using a Monte Carlo method,  which is more computationally expensive as opposed to being on energy and using the Hurd-Zhou method. The models are also different. The Hurd-Zhou method has  also been studied for spread options on equities and  interest rates in Alfeus and  Schl\"{o}gl \cite{AlSc2018}.

    For the WVAG BDLP, we provide simulated data examples under both the true model and the fitted model. Under the true model, the spread option price is computed via the analytical formula for the cgf of the log return with the FFT method, and is compared to the Monte Carlo method, showing close agreement.  We also provide an end-to-end implementation of the spread option pricing in a simulation study by fitting the model.

   As we have specified parametric models for both the energy price process $\bfS$ under the  real world measure $\PP$ and the Esscher measure $\QQ_{\bfh}$, it follows that our model allows the seasonality and LDOUP parameters to be estimated under $\PP$, specifically using the MLE method by Lu \cite{Lu21} combined with the analytical formula for the cgf of the innovation term derived here,  and then allows the  MPR $\bfh$ to be calibrated to observed forward or call prices.  Thus, we have a joint estimation and calibration method that coherently leverages both data on the $\PP$-dynamics of the underlying and its derivative prices to fit the model. This fitting method is consistent with a large stream of literature in incomplete equity markets \cite{ek95,hub,MiSz17,pasc,rus15,sco} and energy markets \cite{bd15,BenthPir18,bsb04,benthbook,bs14,bv13,gzlw23}. We  extend this approach to the multivariate energy setting with a numerical implementation, which is otherwise lacking in the literature.  In contrast, there is another stream of literature, which we do not adopt, that fully calibrates the model under $\QQ$ in the univariate \cite{gss22c,MCC98,sab2020,sab2023,sco} and multivariate \cite{ BB14,gs22,gss22b,gss2021,LS10} settings. While this works in the univariate case, since it effectively amounts to defining and calibrating the model under $\QQ$ without any reference to $\PP$, an issue arises in the multivariate case as joint parameters cannot be calibrated in this way given that multi-asset derivatives are not liquidly traded. Consequently, the joint parameters are instead estimated by calibrating to the historical log return correlation under $\PP$, which effectively amounts to making an unjustified assumption that the $\PP$ and $\QQ$ correlations are equal. This is further discussed in Remark \ref{mcmmrem}. Our method is not faced with this problem.

    Then the fitted model can be used to price spread options over a panel of strikes. We find in our example that calibrating to calls is better than calibrating to forwards, and demonstrate the extent to which errors in estimating the model parameters propagate to pricing errors.
    
    Lastly, we provide a real data example in which the model with WVAG BDLP is fit to Australian data and used to price spread options. In particular, we use electricity prices from two states, New South Wales (NSW) and Victoria, to estimate the model and the corresponding base load quarter futures prices to calibrate it, and find that the overall model has a moderately good fit based on the distribution of the innovation term and the calibration error.

    The paper is organized as follows. Section \ref{prelim} reviews the preliminaries, including notation, definitions of the VG and WVAG processes, and useful properties of LDOUPs. Section \ref{analformsec} presents the main probabilistic result, giving the analytical formula for the cgf of the innovation term of a LDOUP where the BDLP is multivariate VG. Section \ref{energydsec} focuses on the pricing of energy derivatives and is the main part of the paper, where we specify the model, show how the BDLP changes under the Esscher transform, obtain the analytical cgf of the log return, and finally apply this to derive pricing formulas for forwards, calls, and spread options. Section \ref{imple} details the implementation details for both the simulated and real data examples, including simulation, estimation, calibration, and the pricing of spread options using the FFT method. Section \ref{ressec} presents simulated data results based on both the true parameters and a simulation study using fitted parameters. Section \ref{realdata} presents the real data results. Section \ref{conc} concludes with some remarks on possible extensions.

    \section{Preliminaries}\label{prelim}    
    We use the convention that $\bfx=(x_1,\dots,x_n)\in\RR^n$  represents a column vector. Let $\bfSigma = (\Sigma_{ij})_{ij=1}^n\in\RR^{n\times n}$ denote an $n\times n$ real matrix. For a multi-index $\alpha$, define $[\cdot]_\alpha$ as extracting the $\alpha$th element. For $A\subseteq\RR^n$, let $\eins_A$ be its indicator function. Let $\DD:=\{\bfx\in\RR^n : \|\bfx\|\leq 1\}$.  For $\bfx,\bfy\in\RR^n$, define the Euclidean inner product  $\skal{\bfx}{\bfy}:=\bfx'\bfy$, and $\|\bfx\|_{\bfSigma} := \sqrt{\skal{\bfx}{\bfSigma \bfx}}$, which is a norm when $\bfSigma$ is positive definite. Let $\bfe_k\in\RR^n$ be the $k$th standard basis vector. Let $I:[0,\infty)\to[0,\infty)$ be the identity function.     The functions $\exp$ and $\log$ are understood to act componentwise  on vectors so, for example, $e^\bfx = (e^{x_1},\dots, e^{x_n})$.

    Throughout, we work on the stochastic basis $(\Omega,\FFF, \FF,\PP)$ with  filtration $\FF=(\FFF_t)_{t\geq0}$, where $\PP$ represents the real world measure. For now, let $\QQ$ be some probability measure equivalent to $\PP$, which represents a risk-neutral measure.
    
    We write a process $\bfX$  in terms of its marginal components and time marginals by $\bfX=(X_1,\dots,X_n)=(\bfX(t))_{t\ge 0}$. Let $\eqd$ denote equality in distribution or law.  Let $\kappa_{\bfX}$, $\Psi_{\bfX}$ and $\Phi_{\bfX}$ denote respectively the    cumulant  generating function (cgf), characteristic exponent and characteristic function of the random vector $\bfX$, or of $\bfX(1)$ if $\bfX$ is a process, under $\PP$. Specifically, for a random vector $\bfX$, these are defined by
    \begin{align*}
        \kappa_{\bfX}(\bftheta) := \log \EE[e^{\skal{\bftheta}{\bfX}}], \quad \bftheta\in\DDD_{\bfX},
    \end{align*}
    where $\DDD_{\bfX} := \{\bftheta\in\RR^n  : \EE[e^{\skal{\bftheta}{\bfX}}] <\infty \}$ is the domain where $\bfX$ has finite exponential moments, or equivalently, where $\kappa_{\bfX}$ is finite, and
    \begin{align*}
        \Phi_{\bfX}(\bftheta) &:= \EE[e^{\rmi \skal{\bftheta}{\bfX}}] = e^{\Psi_{\bfX}(\bftheta)} , \quad \bftheta\in\RR^n.
    \end{align*}
    Let $\kappa_{\bfX}^\QQ$, $\DDD^{\QQ}_\bfX$, $\Phi_{\bfX}^\QQ$ and $\Psi_{\bfX}^\QQ$  denote the respective  functions and domain under $\QQ$. Let $\bfB\sim BM^n(\bfmu,\bfSigma)$ be an $n$-dimensional Brownian motion with drift $\bfmu\in\RR^n$ and  covariance $\bfSigma\in \RR^{n\times n}$. For $\bft=(t_1,\dots, t_n)\in[0,\infty)^n$, define
    \begin{align*}
        \bfmu\tr \bft := (\mu_1 t_1,\dots, \mu_nt_n),\quad \bfSigma \tr \bft  := (\Sigma_{ij}(t_i\wedge t_j))_{ij=1}^n,
    \end{align*}
    which are the mean vector and covariance matrix of the componentwise composition $\bfB(\bft)$.

    \subsection{Variance gamma and related process}

    We begin with a summary of relevant definitions and results from the studies \cite{BKMS16,BLM17a}. Let $\Gamma_S(a,b)$ denote a univariate gamma subordinator with shape $a>0$ and rate $b>0$. An $n$-dimensional  $\bfV\sim VG^n(b,\bfmu,\bfSigma,\bfeta)$ is a \emph{variance gamma (VG) process} is defined by
    \begin{align*}
        \bfV(t) = \bfeta t + \bfB(G(t)), \quad t\geq 0,
    \end{align*}
    where $\bfB\sim BM^n(\bfmu,\bfSigma)$ and $G\sim\Gamma_S(b,b)$ are independent, and  $\bfeta\in\RR^n$. The parameter $\bfeta$ adds a drift. Whenever a process is written without a drift, such as $VG^n(b,\bfmu,\bfSigma)$, it is assumed that $\bfeta=\bfnull$.  The VG process is a L\'evy process with cgf
    \begin{align}\label{vgcgf}
        \kappa_{\bfV}(\bftheta)=\skal{\bfeta}{\bftheta}-b\log K_{\bftheta}  ,\quad \bftheta\in\DDD_{\bfV},
    \end{align}
    where
    \begin{align}
        K_{\bftheta} &:= 1-\frac{\skal{\bfmu}\bftheta}{b}-\frac{\|\bftheta\|^2_{\bfSigma}}{2b},\label{defnm}\\
        \DDD_{\bfV} &:= \{\bftheta\in\RR^n  \,:\, K_{\bftheta}>0 \}.\label{defnd}
    \end{align} 
    It has   characteristic exponent    $\Psi_{\bfV}(\bftheta) =  \kappa_{\bfV}(\rmi \bftheta)$, $\bftheta \in \RR^n$.

    A \emph{sum of  VG  processes}  (or {VG convolution (VGC)})  $\bfZ\sim {VGC}^n_m(\bfb,\bfmu,\bfSigma,\bfeta)$ is defined by
    \begin{align}
        \bfZ(t) =  \bfeta t + \bfV_1(t)+\dots+\bfV_m(t), \quad t\geq 0, \label{VGCpar}
    \end{align}
    where    $\bfV_j\sim VG^n(b_j,\bfmu_j,\bfSigma_j)$, $j=1,\dots,m$, are  mutually independent VG processes with $\bfb := (b_1,\dots, b_m)$,  $\bfmu:= (\bfmu_1,\dots, \bfmu_m)$,  $\bfSigma := (\bfSigma_1,\dots, \bfSigma_m)$
    and  $\bfeta\in\RR^n$. Clearly $\bfZ$ is still a L\'evy process. 
    
    Let $n\geq 2$, $a> 0$, $\bfalpha=(\alpha_1,\dots,\alpha_n) \in(0,1/a)^n$, $\beta_k:= (1-a\alpha_k)/{\alpha_k}$, $k=1,\dots,n$, and $\bfeta\in\RR^n$.  A \emph{weak variance alpha-gamma (WVAG)} process $\bfZ\sim WVAG^n(a,\bfalpha,\bfmu,\bfSigma,\bfeta)$ is defined by
    \begin{align}\label{wvagpropb}
        \bfZ(t) = \bfeta t+ \bfV_0(t)+(V_1(t),\dots,V_n(t)),\quad t\geq 0,
    \end{align}
    where $\bfV_0\sim VG^n( a,a\bfmu\tr\bfalpha,a\bfSigma\tr\bfalpha)$ and  $V_k\sim  VG^1(\beta_k,\allowbreak\alpha_k\beta_k\mu_k,\alpha_k\beta_k\Sigma_{kk})$, $k=1,\dots,n$, are independent.     This representation allows  us to see that the WVAG process is a sum of VG processes with   $\bfZ \sim VGC^n_{n+1}(\bfb_C,\bfmu_C,\bfSigma_C,\bfeta)$, where
    \begin{align}
        \bfb_C &= (a, \beta_1, \dots, \beta_n ),  \label{wvagVGC1}\\
        \bfmu_C &= (a\bfmu\tr\bfalpha, \alpha_1 \beta_1 \mu_1 \bfe_1, \dots, \alpha_n \beta_n \mu_n \bfe_n ), \label{wvagVGC2} \\
        \bfSigma_C &= (a\bfSigma\tr\bfalpha, \alpha_1 \beta_1 \Sigma_{11} \bfe_1 \bfe_1', \dots, \alpha_n \beta_n \Sigma_{nn} \bfe_n \bfe_n' ), \label{wvagVGC3}
    \end{align}
    which  in turn allows the process to be simulated.      The  marginal components of $\bfZ=(Z_1,\dots, Z_n)$ are $Z_k \sim VG^1(1/\alpha_k,\mu_k,\allowbreak\Sigma_{kk},\eta_k)$, which is a general univariate VG process.

    It turns out that the original  definition of  WVAG process is not \eqref{wvagpropb} but rather   $\bfZ \eqd \bfB\odot \bfT$,
    where $\bfB\sim BM^n(\bfmu,\bfSigma)$, $\bfT$ is an $n$-dimensional alpha-gamma subordinator, and $\odot$ is the weak subordination operation (see \cite{BLM17a}), which gives it a  multivariate time-change interpretation that is otherwise unclear in \eqref{wvagpropb}. Moreover, the original definition is equivalent to \eqref{wvagpropb} for $\bfeta=\bfnull$.

    \subsection{L\'evy-Driven Ornstein-Uhlenbeck Processes}

    Let $\bfZ\sim L^n(\bfmu,\bfSigma,\ZZZ)$ denote an $n$-dimensional L\'evy process with characteristic triplet $(\bfmu,\bfSigma,\ZZZ)$, where $\bfmu\in\RR^n$, $\bfSigma\in\RR^{n\times n}$ is a covariance matrix, and $\ZZZ$ is a L\'evy measure.  With reference to \cite{Lu21,mas04,sato99,saya}, we recall that a L\'evy-driven  Ornstein-Uhlenbeck $\bfX \sim {OU}\text{-}\bfZ(\lambda)$ is defined by the SDE \eqref{ldoup} for $t\geq 0$, where $\bfZ\sim L^n(\bfmu,\bfSigma,\ZZZ)$ is the BDLP, $\lambda>0$, and $\bfX_0$ is independent of $\bfZ$.  Note that there is no loss in generality in having the drift term $-\lambda\bfX(t)\rmd t$ instead of $\lambda(\overline{\bfmu}-\bfX(t))\rmd t$ for some $\overline{\bfmu}\in\RR^n$ since the stationary mean of $\bfX$ is controlled by $\EE[\bfZ(1)]$. There is also no loss in generality in having the L\'evy term $\rmd \bfZ(\lambda t)$ instead of $\rmd \wt \bfZ(t)$ for some   general L\'evy  process $\wt\bfZ\sim L^n(\wt\bfmu,\wt\bfSigma,\wt\ZZZ)$ since we can write $\wt\bfZ\eqd\bfZ\circ (\lambda I)$, where $\bfZ\sim L^n(\wt\bfmu/\lambda,\wt\bfSigma/\lambda,\wt\ZZZ/\lambda)$.

    The solution to SDE \eqref{ldoup} is 
        \begin{align*}
            \bfX(t)=e^{-\lambda  (t-s)} \bfX(s) +e^{-\lambda (t-s)}\int_s^t e^{\lambda (u-s)}\,\rmd\bfZ(\lambda u),\quad  0 \leq s\leq t.
        \end{align*}
        For $0 \leq s\leq t$, we follow the notations of in \cite{vstt} and define the random vectors
        \begin{align*}
            \bfZ^*(s,t ):= \int_s^{t }e^{\lambda (u-s)}\,\rmd\bfZ(\lambda u), \quad \bfZ^*(t):=  \bfZ^*(0,t).
    \end{align*}
    Note that 
    \begin{align}
        \bfZ^*(t) \eqd    \bfZ^*(s,s+t) \eqd \int_0^{\lambda t}e^{ u}\,\rmd\bfZ( u) \label{statincr}
    \end{align}
    for $s\geq 0$,  $t\geq 0$, and $\bfZ^*(t)$ depends on $\lambda$.   It then follows that equally spaced observations of the LDOUP at times $t_i=i\Delta$, $i=0,1,\dots,m$, with sampling interval $\Delta>0$, forms an AR(1) process given by
    \begin{align}
        \bfX(t_{i})= b\bfX(t_{i-1})+\bfZ_{b}^{(i)},\quad i=1,\dots,m,\label{sim}
    \end{align}
    where $b=e^{-\lambda\Delta}$, and $\bfZ_{b}^{(i)} \eqd e^{-\lambda\Delta}\bfZ^*(\Delta)$ are the iid innovations. For convenience,  we instead call $\bfZ^*(\Delta)$ the \emph{innovation term}, and it has a characteristic exponent given by
    \begin{align}\label{cez}
        \Psi_{\bfZ^*(t)}(\bftheta) = \int_{0}^{\lambda t} \Psi_\bfZ (e^{ u}\bftheta)\,\rmd u, \quad \bftheta\in\RR^n
    \end{align}     (see \cite[Theorem 2.2]{saya} from which \eqref{statincr}  also follows).
    This implies that if $\bfZ(t)=\bfeta t + \bfZ_1(t)+\dots+ \bfZ_m(t)$, $t\geq 0$, where $\bfZ_1,\dots, \bfZ_m$ are independent L\'evy processes, then
    \begin{align}
        \Psi_{\bfZ^*(t)}(\bftheta)  = \rmi \skal{(e^{\lambda t}-1)\bfeta}{\bftheta} +\sum_{j=1}^m \Psi_{\bfZ_j^*(t)}(\bftheta). \label{wvaginnovce}
    \end{align}

    We will use the VGC and WVAG processes as the BDLP. Specifically, we define  a multivariate LDOUP $\bfX\sim OU\text{-}VGC^n_m(\lambda,\bfb,\bfmu,\bfSigma,\bfeta)$  as an  \emph{OU\text{-}VGC process} if it has speed of mean reversion $\lambda$ and BLDP $\bfZ\sim VGC^n_m(\bfb,\bfmu,\bfSigma,\bfeta)$. In a similar way, the \emph{OU\text{-}VG process} and \emph{OU\text{-}WVAG process} are defined as having BDLPs  $\bfZ\sim VG^n(b,\bfmu,\bfSigma,\bfeta)$ and $\bfZ\sim WVAG^n(a,\bfalpha,\bfmu,\bfSigma,\bfeta)$, respectively. Since the marginal component of a multivariate LDOUP is a univariate LDOUP, it follows that if  $\bfX\sim OU\text{-}WVAG^n(\lambda, a,\bfalpha,\bfmu,\bfSigma,\bfeta)$, then it has  $k$th marginal component $X_k\sim  OU\text{-}VG^1(\lambda,1/\alpha_k,\mu_k,\Sigma_{kk},\eta_{k})$.

    \section{Analytical formula for cgf of the innovation term}\label{analformsec}

    Our main result is the analytical formula for the cgf of $\bfV^*(t)$, $t>0$, the innovation term  of a multivariate OU-VG process.
    
    This formula, which is stated in Theorem \ref{innovcgf} and will later be applied under both $\PP$ and $\QQ_{\bfh}$, has 3 applications in the pricing of energy derivatives. Firstly, in the univariate and bivariate setting under $\PP$, it is used in the likelihood function to estimate the LDOUP parameters (see Section \ref{estsec}). Secondly, in the univariate setting under $\QQ_{\bfh}$, it is used to price forwards (see Section \ref{forsec}), calls (see Section \ref{callsec}), and calibrate the model (see Section \ref{calsec}). Thirdly, in the bivariate setting under $\QQ_{\bfh}$, it is used to price spread options  (see Sections \ref{spreadsec} and \ref{pricesec}).
    
    To begin, for a general $\bfX\sim OU\text{-}\bfZ(\lambda)$, $\lambda>0$, we give a lemma to deal with the finiteness of the cgf, which is a version of \cite[Lemma 4.1]{bsb04} adapted to our setting with a slightly different proof.

    \begin{lemma}\label{lemma}
        Let $\bfZ\sim L^n(\bfmu,\bfSigma,\ZZZ)$ and $t>0$. Let $\DDD_{\bfZ^*(t)} = e^{-\lambda t}\DDD_{\bfZ}$. If $\bftheta\in \DDD_{\bfZ^*(t)}$, then
        \begin{align}
            \kappa_{\bfZ^*(t)}(\bftheta)=   \int_0^{\lambda t} \kappa_{\bfZ}(e^s\bftheta)\,\rmd s \label{cfgstar}
        \end{align}
        is finite, and also $\kappa_{\bfZ}(e^s\bftheta)$ is finite for all $s\in[0,\lambda t]$.  Furthermore,  \eqref{cfgstar} holds for all  $\bftheta\in\CC^n$ such that $\name{Re}(\bftheta)\in\DDD_{\bfZ}$.
    \end{lemma}
    
    \begin{proof}
        The random vector $\bfZ^*(t)$ is infinitely divisible, and by using the L\'evy measure of $e^{-\lambda t} \wt \bfZ^*(t)$, where $\wt \bfZ = \bfZ\circ (\lambda I)$, given by \cite[Lemma 17.1]{sato99}, it follows that for all Borel sets $B\in\BBB(\RR^n\backslash \{\bfnull\})$, the L\'evy measure of  $\bfZ^*(t)$  is
        \begin{align*}
            \ZZZ^*(B) = \int_{\RR^n\backslash \{\bfnull\} } \int_0^{t} \lambda \eins_{e^{\lambda t}B}(e^{-\lambda s}\bfz)  \, \rmd s   \ZZZ(\rmd \bfz)
        \end{align*} 
        using \cite[Propsotion 11.10]{sato99}.        
        
        Let  $ \bftheta \in \DDD_{\bfZ^*(t)}$, or equivalently $e^{\lambda t}\bftheta\in\DDD_{\bfZ}$.  By the change of variables $u = \lambda (t-s)$ and Tonelli's theorem,
        \begin{align}
            J:= \int_{\DD^C} e^{\skal{\bftheta}{\bfz}} \,  \ZZZ^*(\rmd \bfz) &= \int_{\DD^C} \int_0^{\lambda t} e^{\skal{\bftheta}{e^{u}\bfz}} \, \rmd u   \ZZZ(\rmd \bfz) \nonumber\\
            &=  \int_0^{\lambda t}  \int_{\DD^C} e^{ \skal{e^{u}\bftheta}{\bfz}} \,   \ZZZ(\rmd \bfz) \rmd u.\label{explode}
        \end{align} 
        By the exponential moment theorem (see \cite[Theorem 25.17]{sato99}),  $\DDD_\bfZ$ is a convex set containing $\bfnull$, so that  $e^{\lambda t} \bftheta\in\DDD_{\bfZ}$ implies $e^{ u} \bftheta\in\DDD_{\bfZ}$ for all $u \in[0,\lambda t]$, and hence
        \begin{align*}
            f(u):= \int_{\DD^C} e^{ \skal{e^{u}\bftheta}{\bfz}} \,   \ZZZ(\rmd \bfz)<\infty
        \end{align*}
        for all $u \in[0,\lambda t]$. Next, by a version of the dominated convergence theorem (see \cite[Lemma 16.1]{Ba01}) using $\bfz \mapsto 1+e^{\skal{e^{\lambda t}\bftheta}{\bfz}}$ as the dominating function and the definition of the L\'evy measure for $\ZZZ$, it follows that $f$ is continuous on $[0,\lambda T]$. Thus, $J<\infty$, so that $\kappa_{\bfZ^*(t)}(\bftheta)$ is finite for all $\bftheta\in\DDD_{\bfZ^*(t)}$ by the exponential moment theorem.

        Lastly,  by infinite divisibility, \eqref{cez} has L\'evy-Khintchine  representation that can now be analytically continued by the exponential moment theorem so that  \eqref{cfgstar} holds for all  $\bftheta\in \CC^n$ such that $\name{Re}(\bftheta) \in \DDD_{\bfZ^*(t)}$.
    \end{proof}
    
    \begin{remark}
        Note that $\bftheta\in \DDD_{\bfZ^*(t)}$ is equivalent to $e^{\lambda t}\bftheta\in\DDD_{\bfZ}$. If this holds, then $\kappa_{\bfZ^*(t)}(\bftheta)$ is finite. Following  the proof and considering \eqref{explode}, we see that if  $e^{\lambda t}\bftheta\notin \name {cl}({\DDD_{\bfZ}})$, where $\name{cl}$ is the closure, then $\kappa_{\bfZ^*(t)}(\bftheta)$ is not finite, but if $e^{\lambda t}\bftheta\in \partial({\DDD_{\bfZ}})$, where $\partial$ is the boundary, it is unclear whether $\kappa_{\bfZ^*(t)}(\bftheta)$ is finite or not, as only one endpoint of the integrand is infinite. Accordingly, all sets of the form $\DDD_{\bfZ^*(t)}$ should be understood in the sense that $\DDD_{\bfZ^*(t)}\subseteq \{\bftheta\in\RR^n:\EE[e^{\skal{\bftheta}{\bfZ^*(t)}}] <\infty \}$. \halmos
    \end{remark}

    Let  $\name{Li}_2$ denote the dilogarithm function \cite{zagier}, which is defined by
    \begin{align*}
        \operatorname{Li_2}(z) = \int_0^z -\frac{\log(1-y)}{y}\,\rmd y,
    \end{align*}
    where $z\in\CC\backslash [1,\infty)$. Provided that $z\leq1$,  $\name{Li}_2(z)\in\RR$ with $\name{Li}_2(1)=\pi^2/6$. Now we give the analytical formula for the cgf of $\bfV^*(t)$ using a method inspired by \cite{sab2020}.

    \begin{theorem}\label{innovcgf}   Let $\bfV\sim VG^n(b,\bfmu,\bfSigma)$ and $t>0$. Let   $\DDD_{\bfV^*(t)} = e^{-\lambda t}\DDD_{\bfV}$. If $ \bftheta\in \DDD_{\bfV^*(t)}$, then
        \begin{align}
            { \kappa_{\bfV^*(t)}(\bftheta)=}\begin{cases}
                \begin{aligned}
                    & b[\name{Li}_2(e^{\lambda t} A(\bftheta)) - \name{Li}_2( A(\bftheta)) \\
                    & \phantom{b[}+ \name{Li}_2(-e^{\lambda t} B(\bftheta)) - \name{Li}_2( -B(\bftheta)) ],
                \end{aligned}  & \text{if $\|\bftheta\|^2_\bfSigma >0, $}  \\
                b(\name{Li}_2(e^{\lambda t} \wt A(\bftheta)) - \name{Li}_2( \wt A(\bftheta)), & \text{if $\|\bftheta\|^2_\bfSigma =0 $,} \end{cases}\label{cfgvgg} 
        \end{align}
        where
        \begin{align*}
            A(\bftheta) = \frac{ \sqrt{ \skal{\bfmu} {\bftheta}^2+ 2b  \|\bftheta\|^2_\bfSigma} + \skal{\bfmu} {\bftheta}  }{2b},   \quad B(\bftheta) = \frac{  \sqrt{ \skal{\bfmu} {\bftheta}^2+ 2b  \|\bftheta\|^2_\bfSigma} - \skal{\bfmu} {\bftheta} }{2b},
        \end{align*}
        and   $\wt A(\bftheta) = \skal{\bfmu}{\bftheta}/b$. Furthermore, \eqref{cfgvgg} holds for all $\bftheta\in \CC^n$ such that $\name{Re}(\bftheta) \in \DDD_{\bfZ^*(t)}$.
    \end{theorem}

    \begin{proof}  Let  $\bftheta\in \DDD_{\bfV^*(t)}$. Consider the $\|\bftheta\|_{\bfSigma}^2>0$ case. Let $ A:=  A(\bftheta)$, $B:= B(\bftheta)$, $c:=e^{\lambda t}$. Notice that
        \begin{align*}
            M(u):= 1- \frac{\skal{\bfmu} {\bftheta}}{b} u   -\frac { \|\bftheta\|^2_\bfSigma}{2b}u^2 = (1-Au)(1+Bu).
        \end{align*}

        By Lemma \ref{lemma}, $ \kappa_{\bfV^*(t)}(\bftheta)$ is finite and
        \begin{align}
            \kappa_{\bfV^*(t)}(\bftheta)   &=-b \int_1^{c} \frac{ \log  M(u)}{u}  \, \rmd u \label{thmeq1}\\
            & =-b \int_1^{c} \frac{ \log (1-Au)} {u} +\frac{ \log (1+Bu)} {u}   \, \rmd u  \label{thmeq2}\\
            & =b \int_A^{cA }- \frac{ \log (1-y)} {y}\, \rmd y  +\int_{-B}^{-cB}-\frac{ \log (1-y)} {y}   \, \rmd y,  \label{thmeq3}
        \end{align}
        which gives \eqref{cfgvgg}.

        We now check that the above expressions are defined. Firstly, according to Lemma \ref{lemma}, $M(e^s) >0 $ for $s\in[0,\lambda t]$, or equivalently $M(u) >0$ for $u\in[1,c]$, so \eqref{thmeq1} is defined. The coefficient of $u^2$ in $M(u)$ satisfies  $AB=\|\bftheta\|_{\bfSigma}^2/(2b)>0$. Clearly $A>0$, so we must have $B>0$. Now either  both $1-Au$ and $1+Bu$ are strictly positive for all $u\in[1,c]$, or strictly negative. However,  $1+Bu < 0 $ is impossible since $u < -1/B< 0$ contradicts $u\geq 1$. Thus, we must have  $1-Au> 0 $ and $1+Bu> 0 $ for  all $u\in[1,c]$, so \eqref{thmeq2} is defined.  Observe that $-cB<  B\leq A< cA<1$, where the last inequality is implied by $1-Au> 0 $ on $u\in[1,c]$.  Thus, we have ensured \eqref{thmeq3} can be evaluated using $\name{Li}_2$ on the subset of the domain where the function is real.
        
        Finally, the $\|\bftheta\|_{\bfSigma}^2=0$ case follows from similar arguments by noting that, with $\wt A:= \wt A(\bftheta)$, we have  $M(u) = 1-\wt A u >0$ on $u\in[1,c]$ and $c \wt A<1$.
    \end{proof}
    
    \begin{remark}
        If $\bftheta=\bfnull$, then $\kappa_{\bfV^*(t)}(\bftheta)=0$. If the covariance matrix $\bfSigma$ is invertible, or equivalently positive definite, and $\bftheta\neq\bfnull$, then only the $\|\bftheta\|^2_\bfSigma >0$ case in \eqref{cfgvgg} holds, and the other case is impossible. Also, $\kappa_{\bfV^*(0)}(\bftheta)=0$, $\bftheta \in\DDD_{\bfV^*(0)}=\RR^n$. Finally, as $\bfnull\in\DDD_{\bfV^*(t)}$, the last part of Theorem \ref{innovcgf} gives the characteristic exponent $\Psi_{\bfV^*(t)}(\bftheta) =  \kappa_{\bfV^*(t)}(\rmi \bftheta)$, $\bftheta\in\RR^n$.  \halmos
    \end{remark}

    \begin{corollary}
        Let $V\sim VG^1(b,\mu,\Sigma)$, $\Sigma>0$, and $t>0$.  Let $\DDD_{V^*(t)} = e^{-\lambda t}\DDD_{V}$. If  $\theta\in \DDD_{V^*(t)}$, then 
        \begin{align}
            \kappa_{V^*(t)}(\theta)= {}& b[\name{Li}_2(e^{\lambda t} A(\theta)) - \name{Li}_2( A(\theta)) + \name{Li}_2(-e^{\lambda t} B(\theta)) - \name{Li}_2( -B(\theta)) ], \label{cgf1d}
        \end{align}
        where
        \begin{align}
            A(\theta) =\frac{ \sqrt{\mu^2+ 2b\Sigma}}{2b}|\theta| + \frac{\mu \theta}{2b}, \quad   B(\theta)  = \frac{ \sqrt{\mu^2+ 2b\Sigma}}{2b}|\theta| - \frac{\mu \theta}{2b}.\label{ab}
        \end{align}
        
    \end{corollary}

    \begin{remark} \label{cgfrem}

        The formula for $\kappa_{V^*(t)}(\theta)$ in univariate case was already obtained in Sabino \cite[Equation (18)]{sab2020}. 
         Sabino's proof is based on the well-known fact that the VG process $V\eqd G^+-G^{-}$  is the difference of two independent gamma subordinators since its  characteristic exponent is $\Psi_V(\theta)  = \Psi_{G^+}(\theta) + \Psi_{G^-}(-\theta)  $, where
        \begin{align*}
            \Psi_{G^{\pm}}(\theta)  = -b\log\left( 1- \rmi \frac{\sqrt{\mu^2 + 2b\Sigma } \theta \pm \mu \theta }{2b}  \right), \quad \theta\in\RR.
        \end{align*}
        This is why Sabino's formula has $|\theta|$ in \eqref{ab} replaced with  $\theta$ but is otherwise the same.   Numerically, this  causes the arguments of $\name{Li}_2$ to be different despite the value  of $\kappa_{V^*(t)}(\theta)$ being the same.
        
        As noted in \cite[Remark 7]{Lu21},  Sabino's proof  did not seem to readily extend as it has no  multivariate analog. Specifically, the multivariate VG process $\bfV$  cannot be written as a difference of multivariate gamma processes as it is not in the class of generalized gamma convolutions provided that $\bfSigma$ is invertible.    This makes  Theorem \ref{innovcgf} quite surprising. In fact, using methods similar to the proof allows us to write $\Psi_\bfV(\bftheta)  = \Psi_{\bf G^+}(\bftheta) + \Psi_{\bf G^-}(-\bftheta) $, where
        \begin{align*}
            \Psi_{\bfG^{\pm}}(\bftheta)  = -b\log\left( 1 - \rmi  \frac{   \sqrt{ \skal{\bfmu} {\bftheta}^2+ 2b  \|\bftheta\|^2_{\bfSigma}} \pm \skal{\bfmu} {\bftheta} }{2b}\right),  \quad \bftheta\in\RR^n,
        \end{align*}
        which raises the question of whether $\bfV$ can be written as the difference of two other L\'evy processes $\bfG^{\pm}$ with characteristic exponents  $\Psi_{\bfG^{\pm}}(\bftheta) $. Unfortunately, the answer is no as $\Psi_{\bfG^{\pm}}(\bftheta) $ is not a valid characteristic exponent. If it were, then the first marginal component $G_1^{+}$ would have characteristic exponent 
        \begin{align*}
            \Psi_{G_1^{+}}(\theta_1)  =  -b\log\left( 1 -  \rmi \frac{\sqrt{\mu_1^2 + 2b\Sigma_{11} } |\theta_1| + \mu_1\theta_1  }{2b}        \right),\quad \theta_1\in\RR,
        \end{align*}
        but $ \Phi_{G_1^{+}}(\theta_1) $ fails to be a characteristic function because it can fail to satisfy the necessary condition $\Phi_{G_1^{+}}(-\theta_1)=\overline{\Phi_{G_1^{+}}(\theta_1)}$. \halmos
    \end{remark}

    \section{Energy derivatives}\label{energydsec}

    In this section, we specify our energy market model using a LDOUP, show how a risk-neutral measure can be obtained by the Esscher transform, and how forward contracts, call options and spread options can be priced.

    \subsection{Energy market model} \label{procsec}
    
    We define a market model on  $(\Omega,\FFF, \FF,\PP)$ with  finite time horizon $T$, filtration $\FF=(\FFF_t)_{t\in [0,T]}$,  and real world measure $\PP$.  
    
    Let $\bfLambda= (\bfLambda(t))_{t\in[0,T]} =(\Lambda_1,\dots,\Lambda_n)$  be a deterministic  \emph{seasonality function}. 
    Let 
    \begin{align*}
        \bfX\sim OU\text{-}VGC^n_m(\lambda,\bfb,\bfmu,\bfSigma,\bfeta)
    \end{align*}
    with initial value $\bfX_0$  in its stationary distribution. Thus, the BDLP is $\bfZ\sim VGC^n_m(\bfb,\bfmu,\bfSigma,\bfeta)$, and we work within this more general setting throughout, which includes $\bfZ\sim WVAG^n(a,\bfalpha,\allowbreak \bfmu,\bfSigma,\bfeta)$ as a subclass. In Section \ref{modelparsec}, we make an additional assumption on $\bfeta$, which is not needed now.

    Recalling \eqref{priceproc}, the $n$-dimensional energy spot price process $\bfS = (\bfS(t))_{t\in[0,T]}$ is defined by the log price process
    \begin{align}
        \bfY(t):= \log \bfS(t) = \boldsymbol{\Lambda}(t) +\bfX(t).\label{basiceqn}
    \end{align}
    We assume $\bfZ$, or equivalently $\bfS$, is adapted to $\FF$. Also, define the  log return over time interval $[s,t]$ as
    \begin{align*}
        \bfR(s,t):= \bfY(t) - \bfY(s).
    \end{align*}    
    In this model, $\bfS$ is not a traded asset. The only traded asset is  the risk-free asset, which has constant risk-free rate $r$.
    
    This is known as a geometric model and ensures that prices are positive. In contrast, an arithmetic model would replace \eqref{basiceqn} with $ \bfS(t) = \boldsymbol{\Lambda}(t) +\bfX(t)$. While electricity often exhibits negative intraday prices, derivative pricing is usually  based on average daily price data (for example, see \cite[Section 3]{bs14}, \cite[Section 3]{bsb04}, \cite[Section 9.3.2]{benthbook}), which are almost always positive.  However, in Section \ref{realimpsec}, we have a real data example with negative prices and discuss a trick that allows a geometric model to still be used.

    \subsection{Esscher transform}

    We now use the Esscher transform to obtain an EMM.  For references on Esscher transforms and its applications to derivative pricing  in equity markets, see \cite{gs94, hub}, \cite[Section 13.5]{pasc}, \cite[Section VII.3c]{shir}, and for applications to energy markets, see \cite{benthbook}.
    
    Let $\bfZ\sim L^n$ be a L\'evy process. The \emph{Esscher measure} $\QQ_{\bfh}$ of $\bfZ$ with \emph{market price of risk (MPR)} (or Esscher parameter) $\bfh=(h_1,\dots, h_n)$  is  defined by the  Radon-Nikodym derivative
    \begin{align}
        {\left.\frac{\rmd \QQ_{\bfh}}{\rmd \PP}\right|}_{\FFF_t} =\frac{e^{\skal{\bfh}{\bfZ(t)}}} {\EE[e^{\skal{\bfh}{\bfZ(t)}}]}.\label{esscher}
    \end{align}
    It is immediate that $\QQ_{\bfh}$ exists if and only if  $\bfh\in  \DDD_\bfZ$ and it is equivalent to $\PP$. The Esscher transform  is structure preserving in the sense that $\bfZ$ is also a L\'evy process under $\QQ_\bfh$.

    The no arbitrage condition that the discounted traded asset prices are $\QQ_{\bfh}$-martingales trivially holds for the risk-free asset, which as we recall is the only traded asset, so that the Esssher transform  provides up to infinitely many  EMMs $\QQ_{\bfh}$ indexed by $\bfh\in  \DDD_\bfZ$, all consistent with no arbitrage. The market chooses a MPR $\bfh\in\DDD_{\bfZ}$ based on risk preferences, which we infer by calibrating the model to observed  forward or call prices.

    We now give some results on the $\QQ_{\bfh}$-dynamics of the BDLP $\bfZ$, which by equivalence determines the $\QQ_{\bfh}$-dynamics of $\bfS$. We begin by showing how the parameters of the multivariate VG process change under the Esscher transform. While the univariate case can be found in \cite[Proposition~4]{hub}, the multivariate case does not appear to be readily available except in  \cite[Theorem 6.3.1]{ant}, and we reproduce this proof below. However, it is also an implication of \cite[Theorem 2.16]{BKMS16}.

    \begin{lemma}\label{vgthm}
        Suppose $\bfV\sim VG^n(b,\bfmu,\bfSigma)$ under $\PP$. If $\bfh\in\DDD_{\bfV}$, then $\bfV \sim VG^n(b_h,\allowbreak \bfmu_\bfh,\bfSigma_\bfh)$ under $\QQ_\bfh$, where
        \begin{align}
            b_h = b, \quad \bfmu_h = \frac{\bfmu + \bfSigma \bf h}{K_\bfh}, \quad   \bfSigma_\bfh = \frac{\bfSigma }{K_\bfh},\label{vgespar}
        \end{align}
        and    $ K_\bfh$is defined in \eqref{defnm}.
    \end{lemma}
    \begin{proof}
        Let $\bfh \in \DDD_{\bfV}$, then there always exists an open neighborhood $U\subseteq\RR^{n}$ about $\bfnull$ such that if $\bftheta\in U$, then  $\bftheta+\bfh\in\DDD_{\bfV}$.  The moment generating function of $\bfV$ under $\QQ_\bfh$ is
        \begin{align*}
            M_\bfV^{\QQ_\bfh}(\bftheta) 
            &=     \EE_{\QQ_\bfh}\left[ e^{\skal{\bftheta}{\bfV(1)}} \frac{e^{\skal{\bfh}{\bfV(1)}} }{\EE [e^{\skal{\bfh}{\bfV(1)}}] }\right]\\
            &=\frac{e^{\kappa_\bfV(\bftheta+\bfh)}} {e^{ \kappa_\bfV(\bftheta)}}\\
            & = {\left(\frac{b - \skal{\bfmu}{\bfh} -  \frac 12  \|\bfh\|^2_{\bfSigma} }{ b - \skal{\bfmu}{\bfh} -  \frac 12  \|\bfh\|^2_{\bfSigma}  - \skal{\bfmu+\bfSigma \bfh}{\bftheta} -  \frac 12  \|\bftheta\|^2_{\bfSigma} }\right)}^b \\
            & ={ \left(\frac{b_h  }{ b_h - \skal{\bfmu_{\bfh}}{\bftheta} -  \frac 12  \|\bftheta\|^2_{\bfSigma_{\bfh}}  }\right)}^b
        \end{align*}
        using \eqref{esscher} in the first line, \eqref{vgcgf} in the third line, and \eqref{vgespar} in the last line. As this coincides with the moment generating function of $VG^n(b_h,\allowbreak \bfmu_\bfh,\bfSigma_\bfh)$ on $U$, the result follows.
    \end{proof}

    Note that $\DDD_{\bfV}$ is defined  in \eqref{defnd}  using the parameters $(b,\bfmu,\bfSigma)$ under $\PP$. This is always the case, including similarly for the proceeding results.
    
    We see that the existence  of the Esscher measure $\QQ_{\bfh}$ depends importantly on $\bfh$ and the kurtosis parameter $b>0$. For any choice of $(b,\bfmu,\bfSigma)$, if $\|\bfh\|$ is sufficiently small, then   $K_\bfh>0$, which is equivalent to $\bfh\in\DDD_{\bfV}$ and $\QQ_{\bfh}$ existing. Alternatively, for any choice of $\bfh$, if $b$ is sufficiently large, so that the kurtosis is sufficiently small, then  $K_\bfh>0$. Otherwise, for an arbitrary choice of $(b,\bfmu,\bfSigma)$ and $\bfh$, $\QQ_{\bfh}$ may not exist.

    Let $p_\bfh$ denote the map sending the $\PP$ parameter to the $\QQ_\bfh$ parameters for the VG process, so  $p_\bfh(b,\mu,\bfSigma) := (b_\bfh,\bfmu_\bfh,\bfSigma_\bfh).$

    \begin{corollary}\label{svgn}
        Suppose $\bfZ\sim VGC^n_m(\bfb,\bfmu,\bfSigma,\bfeta)$ under $\PP$ is written  as 
        \begin{align*}
            \bfZ(t)  =\bfeta t + \bfV_1(t)+\dots+\bfV_m(t),\quad t \geq0.
        \end{align*}
        from \eqref{VGCpar}.   If $\bfh \in\DDD_{\bfZ}$, then this holds under $\QQ_\bfh$ with $\bfZ\sim VGC^n_m(\bfb_\bfh,\bfmu_\bfh,\bfSigma_\bfh,\bfeta)$, where $\bfV_j\sim  VG^n ([\bfb_\bfh]_j,[\bfmu_\bfh]_j,[\bfSigma_\bfh]_j)$, $ j=1,\dots,m$, are independent, and
        \begin{align*}
            ([\bfb_\bfh]_j,[\bfmu_\bfh]_j,[\bfSigma_\bfh]_j) =  p_\bfh(b_j,\bfmu_j,\bfSigma_j),
        \end{align*}
        and
        \begin{align}
            \DDD_{\bfZ} = \bigcap_{j=1}^m \DDD_{\bfV_j}.\label{dint}
        \end{align}
    \end{corollary}
    \begin{proof} Without loss of generality, suppose $\bfeta=\bfnull$. Let $\bfh\in \DDD_{\bfZ}$.
        Combining \cite[Section VII.2c, Theorem 1]{shir}, Kac's theorem, and that $\bfV_1,\dots,\bfV_m$ are independent under $\PP$, it follows that they are also independent under $\QQ_{\bfh}$.  Then applying Lemma  \ref{vgthm} to each $\bfV_1,\dots,\bfV_m$ yields its law under under $\QQ_{\bfh}$.

        Note that \eqref{dint} holds since $\EE[e^{\skal{\bftheta}{ \bfV_1+\dots+\bfV_m}}]=\prod_{j=1}^m\EE[e^{\skal{\bftheta}{\bfV_j}}]$ with the LHS finite if and only if each term in the product is, otherwise both sides are infinite.  \end{proof}

    \begin{corollary}\label{vgn}
        Suppose $\bfZ\sim WVAG^n(a,\bfalpha,\bfmu,\bfSigma,\bfeta)$ under $\PP$ is written as
        \begin{align*}
            \bfZ(t) = \bfeta t + \bfV_0(t) + (V_1(t),\dots,V_n(t)),\quad t\geq0,
        \end{align*}
        from \eqref{wvagpropb}.  If  $\bfh \in\DDD_{\bfZ}$, then  this holds under $\QQ_\bfh$, where $\bfV_0 \sim VG^n(a, \bfmu_\bfh,\bfSigma_\bfh)$, $V_k\sim  VG^1(\beta_k, \mu_{h_k}, \Sigma_{h_k})$, $k=1,\dots, n$, are independent, and
        \begin{align}
            \begin{split}
                \bfmu_\bfh &= a\frac{\bfmu\tr\bfalpha + (\bfSigma\tr\bfalpha) \bf h}{K_\bfh},\quad \bfSigma_\bfh = a\frac{\bfSigma\tr\bfalpha }{K_\bfh},\\
                K_\bfh &= 1-{\skal{\bfmu\tr\bfalpha}\bfh}-\frac{\|\bfh\|^2_{\bfSigma\tr\bfalpha}}{2},   \end{split} \label{mvparh}\\
            \begin{split}
                \mu_{h_k} &= \alpha_k\beta_k\frac{\mu_k + \Sigma_{kk} h_k}{K_{h_k}},\quad \Sigma_{h_k} = \alpha_k\beta_k\frac{\Sigma_{kk} }{K_{h_k}},\\
                K_{h_k} &= 1-\alpha_k\mu_kh_k -   \frac{\alpha_k\Sigma_{kk}h_k^2}{2},\label{uvparh}
            \end{split}
        \end{align}
        and
        \begin{align*}
            \DDD_{\bfZ} =  \DDD_{\bfV_0} \cap \bigcap_{i=1}^n \DDD'_{V_k},
        \end{align*}
        $\DDD'_{V_k} := \{\bftheta =(\theta_1,\dots,\theta_n) \in \RR^n:  \theta_k\in \DDD_{V_k}\}$.
    \end{corollary}
    \begin{proof}
        This follows from Corollary \ref{svgn} with the representation $\bfZ\sim {VGC}^n_{n+1}$  from \eqref{wvagVGC1}--\eqref{wvagVGC3}.
    \end{proof}

    \begin{remark}\label{notwvagrem} From Corollary \ref{vgn}, the $k$th marginal component has law  
        \begin{align}
            Z_k(t) = \eta_k t + {[\bfV_0]}_k(t) + V_k(t),\quad t\geq 0,\label{match0}
        \end{align}
        under $\QQ_\bfh$, where 
        \begin{align}
            {[\bfV_0]}_k \sim VG^1(a, {[\bfmu_\bfh]}_k,{[\bfSigma_\bfh]}_{kk}), \quad V_k\sim  VG^1(\beta_k, \mu_{h_k}, \Sigma_{h_k}) ,\label{match}
        \end{align}
        are independent. Since $Z_1$ is a sum of 2 independent univariate VG processes, it is in general not a univariate VG process.

        One way to see this is to take the example $n=2$, $a=0.5$, $\bfalpha=(1,1)$, $\bfmu=\bfeta=\bfnull$, $\bfSigma$ as the identity matrix, and $\bfh=(1,0.5)\in\DDD_{\bfZ}$. Assume for the purpose of contradiction that $Z_1\sim VG^1(b,\mu,\Sigma)$ under $\QQ_{\bfh}$ for some $(b,\mu,\Sigma)$, which can be obtained by matching the first 3 central moments with those obtained from  \eqref{match}. Then $\nu:=1/b$ can be found in closed form by solving a quadratic equation, giving $\nu =   1.03788$ or $ 3.30906$, where only the former satisfies the parameter constraints including $\Sigma>0$, so that $\name{Kurt}_{\QQ_{\bfh}}(Z_1)=8.97664$. On the other hand, \eqref{match} implies $\name{Kurt}_{\QQ_{\bfh}}(Z_1)=9.07618$, a contradiction.  Thus, since a WVAG process must have VG marginal components, the class of WVAG processes is not closed under the Esscher transform.
        
        Relatedly, it should be noted that  the $\QQ_{h_1}$-dynamics of $Z_1$ from the Esscher transform of $Z_1$ is not necessarily the same as the $\QQ_{\bfh}$-dynamics of $Z_1$ from the Esscher transform of $\bfZ$. For example, if $Z_1\sim VG^1$ under $\PP$, then when $n=1$,  $Z_1\sim VG^1$ under $\QQ_{h_1}$, but this may be false when $n\geq 2$ since $h_2,\dots, h_n$ generally, not just $h_1$, affects the law of $Z_1$  under $\QQ_{\bfh}$.  \halmos
    \end{remark}

    \subsection{Analytical formula for cgf of  the log return}
    
    Using the above results on the Esscher transform, we now obtain analytical formulas for the cgf of the log return in  terms of the analytical formula in Theorem \ref{innovcgf}.
    
    Let the time to maturity be $\tau := T-t$. We write the energy price process in the form $ \bfS(T) = \bfS(t) \exp(\bfR(t,T))$, where the log return over $[t,T]$ is
    \begin{align*}
        \bfR(t,T)= \bfLambda(T) + \bfX(t)e^{-\lambda \tau }-\bfY(t)+\bfZ^*({t,T})e^{-\lambda \tau} 
    \end{align*}
    and $\bfX(t)= \bfLambda(t)-\log\bfS(t)$.     Using fact that $\bfZ^*(t,T)$ is independent of $\FFF_t$, and \eqref{statincr}, the conditional distribution $\bfR(t,T)\given \FFF_{t} $ satisfies 
    \begin{align}
        \bfR(t,T)\given \FFF_t &\eqd   \bfLambda(T) + \bfX(t)e^{-\lambda \tau }-\bfY(t)+\bfZ^*e^{-\lambda \tau} \label{logrettau}
    \end{align}
    for a random vector $\bfZ^*\eqd \bfZ^*(\tau)$ independent of $\FFF_t$.   Let $\kappa_{\bfR(t,T)}^{\QQ_\bfh}(\bftheta):=\linebreak \kappa_{\bfR(t,T)\given \FFF_t}^{\QQ_\bfh}(\bftheta)$ denote the cgf of this conditional distribution.

     For a L\'evy process $\bfZ\sim L^n$, recall the notation $\DDD_{\bfZ}^{\QQ_{\bfh}}:= \{\bftheta \in\RR^n: \linebreak\EE_{\QQ_{\bfh}}[e^{\skal{\bftheta}{\bfZ(1)}}]<\infty\}$. It follows from \cite[Theorem 13.66]{pasc} that $\bftheta\in\DDD_{\bfZ}^{\QQ_{\bfh}}$ is equivalent to $\bftheta+\bfh\in\DDD_{\bfZ}$, which expresses the domain of  $\kappa_{\bfZ}^{\QQ_\bfh}$ under $\QQ_\bfh$ in terms of the domain $\DDD_{\bfZ}$ under $\PP$, which is the approach we mainly use. Alternatively, from Corollary \ref{svgn}, noting that $\bfZ\sim VGC^n_m$ is closed under Esscher transform but with different parameters under $\PP$ and $\QQ_{\bfh}$,  $\DDD_{\bfZ}^{\QQ_{\bfh}}$ can be understood as $\DDD_{\bfZ}$ but with the $\PP$-parameters replaced by the $\QQ_{\bfh}$-parameters. 
   
    \begin{proposition} \label{svgmultichar}     Using the setting and notation of  Corollary \ref{svgn}, suppose $\bfZ \sim VGC^n_m(\bfb,\bfmu,\bfSigma,\bfeta)$ under $\PP$ and $\bfh\in\DDD_{\bfZ}$.  Then
        \begin{align}
            \kappa_{\bfR(t,T)}^{\QQ_\bfh}(\bftheta) = \skal{\bfLambda(T) + \bfX(t) e^{-\lambda \tau} -\bfY(t)} {\bftheta}  +\kappa_{\bfZ^*(\tau)}^{\QQ_\bfh} (e^{-\lambda\tau}\bftheta)\label{eqc}
        \end{align}
        for $\bftheta+\bfh \in \DDD_{\bfZ}$, where
        \begin{align}
            \kappa_{\bfZ^*(\tau)}^{\QQ_\bfh} (\bftheta) =   \skal{(e^{\lambda\tau}-1  )\bfeta}{\bftheta} + \sum_{j=1}^m \kappa_{\bfV_j^*(\tau)}^{\QQ_\bfh} (\bftheta),\label{zstarsvg}
        \end{align}
        and the law of $\bfV_j$, $j=1,\dots,m$, under $\QQ_{\bfh}$  is given in   Corollary \ref{svgn}, and $ \kappa_{\bfV_j^*(\tau)}^{\QQ_\bfh}$ is evaluated by  \eqref{cfgvgg}. Furthermore,  \eqref{eqc} holds for $\bftheta\in\CC^n$ such that $\name{Re}(\bftheta)+\bfh\in\DDD_{\bfZ}$. 
    \end{proposition}
    
    \begin{proof}
        By \eqref{logrettau}, we have
        \begin{align*}
            \EE_{\QQ_h}[e^{\skal{\bftheta}{\bfR(t,T)}}\given \FFF_t]  =    e^{  \skal{\bfLambda(T) + \bfX(t)e^{-\lambda \tau }-\bfY(t)}{\bftheta} }        \EE_{\QQ_h}[e^{\skal{e^{-\lambda \tau }\bftheta }{\bfZ^*(\tau)}}],
        \end{align*} 
        which is finite since the assumption $\bftheta+\bfh \in \DDD_{\bfZ}$ is equivalent to $\bftheta\in\DDD_{\bfZ}^{\QQ_\bfh}$, which is equivalent to $e^{-\lambda \tau}\bftheta\in\DDD_{\bfZ^*(\tau)}^{\QQ_\bfh}$. Thus, we obtain \eqref{eqc}. Next, \eqref{zstarsvg} follows from   Corollary \ref{svgn}, which gives the law of $\bfZ$ under $\QQ_{\bfh}$, and \eqref{wvaginnovce} which gives the cgf of the corresponding $\bfZ^*(\tau)$. Lastly, $\kappa_{\bfR(t,T)}^{\QQ_\bfh}$ can be analytically  continued by the exponential moment theorem  for all $\bftheta\in \CC^n$ such that $\name{Re}(\bftheta) +\bfh\in \DDD_{\bfZ}$.
    \end{proof}
    \begin{corollary} \label{wvagcor}  Using the setting and notation of  Corollary \ref{vgn}, suppose $\bfZ\sim WVAG^n(a,\bfalpha,\bfmu,\bfSigma,\bfeta)$ under $\PP$ and $\bfh\in\DDD_{\bfZ}$.  Then  $\kappa_{\bfR(t,T)}^{\QQ_\bfh}(\bftheta)$ follows \eqref{eqc} for $\bftheta +\bfh \in\DDD_{\bfZ}$,
        where
        \begin{align*}
            \kappa_{\bfZ^*(\tau)}^{\QQ_\bfh} (\bftheta) =    \skal{(e^{\lambda\tau}-1  )\bfeta}{\bftheta} + \kappa_{\bfV_0^*(\tau)}^{\QQ_\bfh} (\bftheta) + \sum_{i=1}^n \kappa_{V_k^*(\tau )}^{\QQ_\bfh} (\theta_k)
        \end{align*}
        where the law of $\bfV_0$, $V_k$, $k=1,\dots,n$, under $\QQ_{\bfh}$ as given in   Corollary \ref{vgn}, and $\kappa_{\bfV_0^*(\tau)}^{\QQ_\bfh}$, $\kappa_{V_k^*(\tau)}^{\QQ_\bfh} $ are evaluated by  \eqref{cfgvgg}. Furthermore,  \eqref{eqc} holds for $\bftheta\in\CC^n$ such that $\name{Re}(\bftheta)+\bfh\in\DDD_{\bfZ}$.
    \end{corollary}
    \begin{proof}
        This follows from  Proposition \ref{svgmultichar} and  Corollary \ref{vgn}.
    \end{proof}

    \begin{proposition} \label{svgunichar}     
        Using the setting and notation of  Corollary \ref{svgn}, suppose $\bfZ \sim VGC^n_m(\bfb,\bfmu,\bfSigma,\bfeta)$ under $\PP$ and $\bfh\in\DDD_{\bfZ}$. Then $ V_{jk} := {[\bfV_j]}_k \sim VG^1( {[[\bfb_\bfh]_j]}_k, \allowbreak {[[\bfmu_\bfh]_j]}_k,  {[[\bfSigma_\bfh]_j]}_{kk})$, $j=1,\dots,m$, and 
        \begin{align}
            \kappa^{\QQ_\bfh}_{R_k(t,T)}(\theta) =   (\Lambda_k(T)+ X_k(t) e^{-\lambda \tau }-Y_k(t))\theta   +\kappa^{\QQ_\bfh}_{Z_k^*(\tau)}( e^{-\lambda \tau}\theta) \label{cgf1dQ}
        \end{align}
        for $\theta\in\RR$ such that  $\theta\bfe_k+\bfh\in\DDD_\bfZ$,  where
        \begin{align*}
            \kappa^{\QQ_\bfh}_{Z_k^*(\tau)}(\theta)  =   (e^{\lambda \tau }-1) \eta_k \theta  +\sum_{j=1}^m  \kappa^{\QQ_\bfh}_{V_{jk}^*(\tau )}(\theta),
        \end{align*}
        and $\kappa^{\QQ_\bfh}_{V_{jk}^*(\tau )}$ is evaluated by \eqref{cgf1d}.  Furthermore,  \eqref{cgf1dQ} holds for $\theta\in\CC$ such that $\name{Re}(\theta)\bfe_k+\bfh\in\DDD_{\bfZ}$.
    \end{proposition}
    
    \begin{proof}
        The proof follows from similar arguments as the proof of Proposition \ref{svgmultichar}, noting that $\kappa^{\QQ_\bfh}_{Z_k^*(\tau)}(\theta) = \kappa^{\QQ_\bfh}_{\bfZ^*(\tau)}(\theta\bfe_k)$. 
    \end{proof}
    
    \begin{corollary}\label{wvaglogretcgf}
        Using the setting and notation of  Corollary \ref{vgn}, suppose   $\bfZ\sim WVAG^n(a,\bfalpha,\bfmu,\bfSigma,\bfeta)$  under $\PP$ and $\bfh\in\DDD_{\bfZ}$. Then $ V_{0k} := {[\bfV_0]}_k$ and $V_k$ satisfies \eqref{match}, and $\kappa^{\QQ_\bfh}_{R_k(t,T)}(\theta)$ follows   \eqref{cgf1dQ}   for $\theta\in\RR$ such that $\theta\bfe_k+\bfh\in\DDD_\bfZ$,
        where
        \begin{align*}
            \kappa^{\QQ_\bfh}_{Z_k^*(\tau)}(\theta)  =  (e^{\lambda \tau}-1) \eta_k \theta  +  \kappa^{\QQ_\bfh}_{V_{0k}^*(\tau )}(\theta) +\kappa^{\QQ_\bfh}_{V_k^*(\tau)}(\theta),
        \end{align*}
        and $\kappa^{\QQ_\bfh}_{V_{0k}^*(\tau)}$, $\kappa^{\QQ_\bfh}_{V_k^*(\tau)}$ are evaluated by \eqref{cgf1d}. Furthermore,  \eqref{cgf1dQ} holds for $\theta\in\CC$ such that $\name{Re}(\theta)\bfe_k+\bfh\in\DDD_{\bfZ}$.

    \end{corollary}
    \begin{proof}
        This follows from  Proposition \ref{svgunichar} and  Corollary \ref{vgn}.
    \end{proof}
    
    \begin{remark}\label{simpdcond}
        Define the interval     \begin{align*}
            D(b,\mu,\Sigma):=            \left( -  \sqrt{  \frac{ \mu^2}{\Sigma^2} + \frac{2b  }{\Sigma} }  -  \frac{ \mu}{\Sigma}  , \sqrt{  \frac{ \mu^2}{\Sigma^2} + \frac{2b  }{\Sigma} }  -  \frac{ \mu}{\Sigma}\right).
        \end{align*}
        Using the notation in Corollary \ref{vgn}, the condition $\theta\bfe_k+\bfh\in\DDD_\bfZ$  in Corollary \ref{wvaglogretcgf} can be written explicitly as an interval. Specifically, it is equivalent to $\theta\bfe_k\in\DDD_\bfZ^{\QQ_\bfh}$,  which is equivalent to
        \begin{align*}
            \theta \in D(a,{[\bfmu_\bfh]}_k,{[\bfSigma_\bfh]}_{kk})\cap D(\beta_k,\mu_{h_k},\Sigma_{h_k}).
        \end{align*}
        \halmos
    \end{remark}
    
    After all the substitutions, all the cgfs  $\kappa^{\QQ_\bfh}_{\bfR(t,T)} $, $\kappa^{\QQ_\bfh}_{R_k(t,T)} $ presented in  this section have analytical formulas in terms of the dilogarithm function.
    
    \subsection{Pricing energy derivatives}\label{derivsec}

    We now use the analytical formulas for the cgf of the log return from the previous section to  obtain pricing formulas for forwards, calls, and spread options.
    
    The current time is $t$, the maturity time of the derivative is $T$ and the time to maturity is $\tau = T-t$. Recall that $\bfZ \sim VGC^n_m(\bfb,\bfmu,\bfSigma,\bfeta)$ or $\bfZ\sim WVAG^n(a,\bfalpha,\allowbreak\bfmu,\bfSigma,\bfeta)$.

    \subsubsection{Forward prices}\label{forsec}
    
    Let $S_k$ be a univariate energy price which is  the $k$th marginal component of $\bfS$. The following result gives an analytical formula for its forward price.

    \begin{proposition}    
        Suppose that $\bfZ \sim VGC^n_m(\bfb,\bfmu,\bfSigma,\bfeta)$ or $\bfZ\sim WVAG^n(a,\bfalpha,\allowbreak\bfmu,\bfSigma,\bfeta)$   under $\PP$, and assume $\bfh,\bfe_k+\bfh \in\DDD_{\bfZ}$.  The  forward price of $S_k$ at time $t$ with maturity time $T$ is
        \begin{align}
            F_k(t,T) &=   S_k(t) \exp ({\kappa_{R_k(t,T)}^{\QQ_\bfh}(1) }), \quad t\in[0,T], \label{fwdpr}
        \end{align}
        where $\kappa_{R_k(t,T)}^{\QQ_\bfh}$ and $\DDD_{\bfZ}$ are given by Proposition \ref{svgunichar} or Corollary \ref{wvaglogretcgf}, respectively. 
    \end{proposition}
    
    \begin{proof} As   $\bfh \in\DDD_{\bfZ}$, the Esscher measure $\QQ_{\bfh}$ exists. Then the forward price is
        \begin{align*}
            F_k(t,T) =  \EE_{\QQ_{\bfh}} [S_k(T)\given \FFF_t] =   S_k(t)  \EE_{\QQ_{\bfh}} [e^{R_k(t,T)}\given \FFF_t].
        \end{align*}
        which gives the result if the expectation is finite, and this is provided by  $\bfe_k+\bfh \in\DDD_{\bfZ}$ by  Proposition \ref{svgunichar} or Corollary \ref{wvaglogretcgf}.
    \end{proof}

    Figure \ref{forwardfig} shows possible shapes of the forward curve, which may increase or decrease in the short run but tend to increase in the long run. This can be contrasted with the standard Black-Scholes model where the forward curve is $\tau\mapsto S_te^{r\tau}$ or to those in \cite[Figure 7]{BenthPir18} for wind power futures.
    
    \begin{figure}[htpb]
        \begin{center}
            \includegraphics{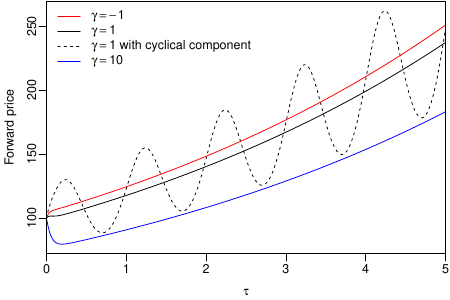}
            \caption{Forward price curve $\tau \mapsto   F_1(8,8+\tau)$, where the MPR is $\bfh = \gamma(-0.1,-0.03)$, $\gamma=-1,1,10$. All parameters are otherwise taken from Section \ref{modelparsec} but  $b_{21}=b_{31}=0$ for all curves except the one with cyclical component.}\label{forwardfig}
        \end{center}
    \end{figure}

     A \emph{base load   futures contract} on the energy price $S_k$ with \emph{contract period} (or delivery period) $(T_1,T_2]$ consisting of $n$ days $t_1,\dots,t_n\in(T_1,T_2]$ and delivery price $K$ pays 
        \begin{align*}
            f_{T_2} =   \frac{1}{n} \sum_{i=1}^nS_k(t_i ) - K
        \end{align*}
        at maturity time $T_2$. Thus,  the base load futures price  $\overline{F}_k(t,T_1,T_2)$ at time $t$, being the delivery price such that the contract have 0 value, is
        \begin{align}
            \overline{F}_k(t,T_1,T_2) =   \frac{1}{n} \sum_{i=1}^n \EE_{\QQ_{\bfh}}[S_k(t_i)\given \FFF_t],\label{futuresprice}
        \end{align}
        where
        \begin{align}
            \EE_{\QQ_{\bfh}}[S_k(t_i)\given \FFF_t] = \begin{cases}
                F_k(t,t_i) & \text{if $t<t_i$,}\\
                S_k(t) & \text{if $t\geq t_i$,}
            \end{cases}\label{futurespricessum}
        \end{align}
        with the forward price $F_k(t,t_i)$ given by \eqref{fwdpr}.

    \subsubsection{Call options}\label{callsec}
    
    Consider a call option on the univariate energy price $S_k$ with strike price $K$ and maturity time $T$.   Suppose that $\bfh, (\epsilon+1)\bfe_k+\bfh\in\DDD_{\bfZ}$, then by the Carr-Madan formula  \cite{CM99}, the  price of the call option at time $t$ is
    \begin{align}
        c_t =C \int_0^\infty \name{Re}\left( \frac{  S_k(t)^{\rmi \theta}K^{-\rmi \theta} \exp(  {\kappa^{\QQ_\bfh}_{R_k(t,T)} }(\epsilon+1 +\rmi \theta))}{ \epsilon^2+\epsilon-\theta^2 + \rmi (2\epsilon+1) \theta}\right) \rmd \theta, \quad t\in[0,T],\label{cm}
    \end{align}
    where
    \begin{align*}
        C = \frac{e^{-r(T-t)}}{\pi} K^{-\epsilon}S_k(t)^{\epsilon+1} ,
    \end{align*}
    and  $\epsilon>0$ is the damping parameter ($\alpha$ in \cite{CM99}). Here, the cgf $\kappa^{\QQ_\bfh}_{R_k(t,T)}(\theta)$ is obtained by Proposition \ref{svgunichar} or Corollary \ref{wvaglogretcgf} and its finiteness in the integrand is provided by $(\epsilon+1)\bfe_k+\bfh\in\DDD_{\bfZ}$ which is equivalent to $(\epsilon+1)\bfe_k\in\DDD_{\bfZ}^{\QQ_{\bfh}}$ . As usual, $\bfh\in\DDD_{\bfZ}$ is necessary for $\QQ_{\bfh}$ to exist. Provided these conditions hold, $\epsilon$ has no effect on the call price.    We make further comments on the choice of $\epsilon$ in Remark \ref{epsrem}.

    \begin{remark}    Note that $(\epsilon+1)\bfe_k\in\DDD_{\bfZ}^{\QQ_{\bfh}} $ is equivalent  to $\EE_{\QQ_\bfh}[S_k^{\epsilon+1}]<\infty$, which is the necessary and sufficient condition for the Carr-Madan formula to be finite. For example, if the BDLP is $\bfZ\sim WVAG^n(a,\bfalpha,\bfmu,\bfSigma,\bfeta)$ under $\PP$, then by Remark \ref{simpdcond}, this condition is equivalent to
        \begin{align*}
            \epsilon \in (0,\infty)\cap(D(a,{[\bfmu_\bfh]}_k,{[\bfSigma_\bfh]}_{kk})-1)\cap (D(\beta_k,\mu_{h_k},\Sigma_{h_k})-1).
        \end{align*}
        \halmos
    \end{remark}

    \subsubsection{Spread options}\label{spreadsec}

    Consider a spread option on the biviarate energy prices $(S_k,S_l)$, $k, l=1,\dots, n$, $k\neq l$, with strike price $K$ and maturity time $T$. By Remark \ref{notwvagrem}, if $n\geq 3$, this is a  more general setting with the possibility that the components $\bfh$ other than the $k$th and $l$th component affect the law of $(S_k,S_l)$ under $\QQ_{\bfh}$.

    Suppose that $\bfh\in\DDD_{\bfZ}$, then by the risk-neutral pricing formula, the price of the spread option at time $t$ is
    \begin{align}
        f_t((S_k,S_l)(t),K) := e^{-r(T-t)}\EE_{\QQ_{\bfh}}[(S_k(T) - S_l(T)  - K)^+\given\FFF_t],\label{spread}
    \end{align}
    $t\in[0,T]$, expressed as a function of the current energy price $(S_k,S_l)(t)$ and strike price $K$, which is possible as the LDOUP is Markov  and the option is path independent.
    Suppose also that $-\bfepsilon+\bfh\in\DDD_{\bfZ}$, then the Hurd-Zhou formula \cite{HZ10} gives  the price with strike price 1  as
    \begin{align}
        \begin{split}
            &  f_t((S_k,S_l)(t),1)= \\
            & \quad\quad \frac{e^{-rT}}{(2\pi)^2} \int_{\RR^2 } \exp(\kappa_{(R_k,R_l)(t,T)}^{\QQ_{\bfh}}(-\bfepsilon +\rmi \bftheta) )e^{\skal{-\bfepsilon+\rmi\bftheta}{ (Y_k,Y_l)(t)}} \hat P(\bftheta+ \rmi \bfepsilon)\,\rmd \bftheta,
        \end{split}\label{spreadint}
    \end{align}
    where 
    \begin{align*}
        \hat P(\bftheta) = \frac{\Gamma(-1+\rmi (\theta_k +\theta_l) ) \Gamma(-\rmi \theta_l)}{\Gamma(1+\rmi \theta_k)}, \quad \bftheta = (\theta_k, \theta_l),    \end{align*}
    $\Gamma$ is the gamma function, and $\bfepsilon = (\epsilon_1,\epsilon_2)$ are damping parameters that satisfy   $\epsilon_1 + \epsilon_2 < -1$, $\epsilon_2>0$.   Here, $\kappa^{\QQ_\bfh}_{(R_k,R_l)(t,T)}(\bftheta)$ can be obtained by  Proposition \ref{svgmultichar} or Corollary \ref{wvagcor} by taking the $k$th and $l$th marginal component if $n\geq3$, similar to Proposition \ref{svgunichar}, and its finiteness in the integrand is provided by   $-\bfepsilon+\bfh\in\DDD_{\bfZ}$, which is equivalent to $-\bfepsilon\in\DDD_{\bfZ}^{\QQ_{\bfh}}$. As before,   provided these conditions hold, $\bfepsilon$ has no effect on the spread option price.
    Then the spread option price for a general strike price $K$ is
    \begin{align}
        f_t((S_k,S_l)(t), K)  = Kf_t\left(\frac{(S_k,S_l)(t)}{K}, 1\right).\label{sprid}
    \end{align}

        \begin{remark} \label{parity}
            For simplicity, let $n=2$. If we have an energy price process $(S'_1,S'_2)$ that is not necessarily positive, for which there exists a constant $\bfc\in\RR^2$ that makes it positive such that $(S_1,S_2)(t)= (S'_2,S'_1)(t)+\bfc$, $t\geq 0$, satisfies the model, then there is a parity relation between spread option prices  on $(S'_1,S'_2)$ and $(S_1,S_2)$. 
            
            Specifically,  consider a spread option on $(S'_1,S'_2)$ with strike price $K'>0$, which pays off
            \begin{align}
                f'_T = (S'_1(T)-S'_2(T)-K')^+\label{pp}
            \end{align}
            at maturity time $T$, let  $c:=c_1+c_2$, and suppose $K:=c-K'>0$, then we have the parity relation
            \begin{align*}
                (S_1(T)-S_2(T)-K)^+ - (S'_1(T)-S'_2(T)-K')^+ = S_1(T)-S_2(T)-K.
            \end{align*}
            It follows that the spread option at time $t$ on  $(S'_1,S'_2)$  with strike price $K'$ is
            \begin{align}
                f'_t((S_1,S_2)(t), K)   =   f_t((S_1,S_2)(t), K)  -e^{-r(T-t)}(F_1(t,T) - F_2(t,T) - K),\label{parrel}
            \end{align}
            where  $f_t((S_1,S_2)(t), K)$ is the spread option price on $(S_1,S_2)$ with strike price $K$ in \eqref{sprid}, which can be computed using the Hurd-Zhou formula as mentioned above, and   $F_k(t,T)$ is the forward prices of $S_k$, $k=1,2$, computed using \eqref{fwdpr}.  To express this in terms of the  forward price of $S_k'$, which is $F'_k(t,T) = \EE_{\QQ_{\bfh}}[S'_k(T)\given\FFF_t]$,  use
            \begin{align}
                F_k(t,T) = F'_l(t,T)+c_k, \quad k\neq l. \label{futadj}
            \end{align}

            Also,  there is no loss in generality in assuming $K>0$, a similar parity can be found if $K<0$. \halmos

        \end{remark}

    \begin{remark}  Throughout Section \ref{derivsec}, the underlying for energy derivatives is the energy price as opposed to the forward or futures prices. However, both forms of underlying appear in the literature (for example, see \cite{bd15,BenthPir18,bsb04,gzlw23} for energy,  \cite{bsb06,bs14,gs22,gss22b} for forwards and futures, and \cite{benthbook,cardur03} for both). By using \eqref{fwdpr}, we can convert from a payoff in terms of forward or futures prices into one in terms of energy prices. \halmos
    \end{remark}
    
    \section{Implementation}\label{imple}

    We now outline for a specific choice of model  how the pricing formulas in Section \ref{derivsec} can be applied to simulated and real data.

    \subsection{Simulated data example}\label{simimpsec}
    
    In the simulated data example, we provide an end-to-end numerical implementation that includes simulating the model, fitting the model, and pricing spread options.

    \subsubsection{Model and parameter specification}\label{modelparsec}
    We begin by specifying the true model, in which the parameters of the first and second marginal components have been  chosen to approximately resemble the  real data example in \cite[Section 5.1]{benthbook} and \cite{bs14}, respectively.
    
    Let $n=2$ and suppose the energy spot price process $\bfS(t)$, $t\in[0,T_o]$ with terminal time $T=T_o$, is as specified in Section \ref{procsec}, where the LDOUP is $\bfX \sim {OU\text{-}WVAG}^2(\lambda,a,\alpha,\bfmu,\bfSigma,\bfeta)$ with BDLP $\bfZ\sim{WVAG}^2(a,\alpha,\bfmu,\bfSigma,\bfeta)$.
    
    Let the seasonality function be
    \begin{align*}
        \Lambda_k(t) = b_{0k} + b_{1k}t + b_{2k}\cos\left(\frac{2\pi}{p}t\right)+b_{3k}\sin\left(\frac{2\pi}{p}t\right), \quad t\in[0,T],
    \end{align*}
    $k=1,2$, where  $p=1$ is the period, and  $\bfb:= (b_{01},\dots,\allowbreak b_{31}, b_{02},\dots,b_{32})\in\RR^8$ are the \emph{seasonality parameters} with values
    \begin{align}
        (b_{01},b_{11},b_{21},b_{31})  &= (3.16132, 0.17500, 0.04385, 0.22986),\label{spars1}\\
        (b_{02},b_{12},b_{22},b_{32})  &= (  2.04056 ,    0.31390 ,    0.01257 ,    0.06587). \label{spars2}
    \end{align}
    The \emph{LDOUP parameters} are $\bfvtheta := (a,\bfalpha,\bfmu,\bfSigma,\bfeta)$ with values    
    \begin{align}
        \begin{gathered}
            \lambda = 18.25,\quad a = 3.15854,\quad\bfalpha =\begin{pmatrix}\alpha_1\\ \alpha_2\end{pmatrix} =\begin{pmatrix}0.17183\\ 0.28494\end{pmatrix}, \\ \bfmu =\begin{pmatrix}\mu_1\\ \mu_2\end{pmatrix}=\begin{pmatrix} -0.03071\\ -0.20335\end{pmatrix} , \quad \bfSigma =  \begin{pmatrix}\Sigma_{11} & \Sigma_{12}\\\Sigma_{12}  &  \Sigma_{22}  \end{pmatrix} = \begin{pmatrix}0.24598 & 0.19452\\0.19452  &  0.17045  \end{pmatrix}, 
        \end{gathered} \label{pars}
    \end{align}
    and $\bfeta =-\bfmu$. We refer to the seasonality and LDOUP parameters together as the \emph{model parameters}.     It is necessary to choose $\bfeta$ such that $\EE[\bfZ(1)] =\overline{\bfmu}=\bfnull$, which is equivalent to $\bfeta =-\bfmu$, so that the model is identifiable, otherwise both $\bfeta$ and the intercept $(b_{01}, b_{02})$ act as location parameters.  The initial value $\bfX_0$ is an independent random vector in the stationary distribution.

    The risk-free rate is $r= 0.05$.  The MPR is
    \begin{align*}
        \bfh = \begin{pmatrix}h_1\\ h_2\end{pmatrix}=\begin{pmatrix} -0.1\\ -0.03\end{pmatrix},
    \end{align*}
    which satisfies the necessary condition $\bfh \in\DDD_{\bfZ}$.
    
    The current time is  $T_e$, and $\bfS$ has been  observed in the time interval $t\in[0,T_e]$ with sampling interval $\Delta$ to estimate the model parameters. At time $T_e$, we price spread options where the strike prices are $K_i$, $i=1,\dots, 13$, which are equally spaced values from 0.4 to 10, the maturity time is $T_o>T_{e}$, and the time to maturity is $\tau = T_o-T_e$. We set $T_{e}=8$, $T_{o}=8.5$, $\Delta = \frac{1}{250}$ with time measured in years. The current energy price is $\bfS(T_e) = (100,96)$

    The Carr-Madan damping parameter for call options on  $S_1$ and $S_2$ is $\epsilon=2.5$, which satisfies the conditions $(\epsilon+1)\bfe_1+\bfh,(\epsilon+1)\bfe_2+\bfh\in\DDD_{\bfZ}$. The Hurd-Zhou damping parameter is $\bfepsilon = (\epsilon_1,\epsilon_2) = (-3.5,1)$, which satisfies the  conditions $\epsilon_1 + \epsilon_2 < -1$, $\epsilon_2>0$ and $-\bfepsilon+\bfh\in\DDD_{\bfZ}$.

    \begin{remark}\label{epsrem}
        In theory, any damping parameter satisfying the above-mention\-ed condition has no effect on the call and spread option prices. In practical terms, when numerically evaluating the pricing formulas, we check that changing to nearby damping parameters has virtually no effect, and any choice satisfying this will do. If it appears no damping parameter achieves this, then for instance, the discretization parameters of the FFT method in Section \ref{pricesec} should be adjusted to obtain a finer grid. Additional discussions on the choice of damping parameter is given in Section \ref{trueparsec1} and  \cite[Section 4.2]{AlSc2018}.\halmos
    \end{remark}

    \subsubsection{Simulation}
    The observations of the LDOUP $\bfX(t_i)$ $t_i=i\Delta$, $i=0,1,\dots, m$, where $m=\lfloor{T_e/\Delta}\rfloor=2000$ are simulated with  $\wt n =1000$ subintervals using the method in \cite[Section 4.2]{Lu21}. This means we use \eqref{sim}, where $\bfZ^*(\Delta )$ is simulated by discretizing the stochastic integral in \eqref{statincr} over $[0,\lambda \Delta] $ into $\wt n$ subintervals. Since the initial value $\bfX_0$ is assumed to be in its stationary distribution, which exists by \cite[Proposition 1(i)]{Lu21}, $\bfX_0$ is simulated by 20\% burn-in.

    \subsubsection{Estimation}\label{estsec}

    By \eqref{basiceqn}, the seasonality parameters  $(b_{0k},b_{1k},b_{2k},b_{3k})$  are estimated with linear regression on the log price $Y_k(t_i)$, $i=0,1,\dots, m$, for each marginal component $k=1,2$. The estimation is valid since the errors are a zero-mean stationary AR(1) process.  Let the estimated seasonality parameters be  $\wh \bfb$ and the estimated seasonality function be $({\wh \Lambda}_1,{\wh\Lambda}_2)$, that is $({\Lambda}_1,{\Lambda}_2)$ with $\bfb= \wh\bfb$, and let
    \begin{align}
        \wh  X_k(t)  =  Y_k(t)-  \wh\Lambda_k(t) \quad k=1,2,\label{filteredobs}
    \end{align}
    so that $(\wh X_1(t_i), \wh X_2(t_i))=\bfx_i$, $i=0,1,\dots,m$, are the filtered observations of the LDOUP used to estimate $\bfvtheta$ with the MLE method in  \cite{Lu21}.
    
    Specifically, the likelihood function we use is $\bfvtheta\mapsto L(\bfvtheta,\bfx_0,\dots,\bfx_m)$, where
    \begin{align}\label{likfn}
        L(\bfvtheta,\bfx_0,\dots,\bfx_m) = e^{mn\lambda\Delta}\prod_{k=1}^m f_{\bfZ^*(\Delta)}(e^{\lambda\Delta}\bfx_k-\bfx_{k-1}).
    \end{align}
    Here, $ f_{\bfZ^*(\Delta)}$  is the pdf of $\bfZ^*(\Delta)$, which is computed by taking the Fourier inversion of the characteristic function  $\Phi_{\bfZ^*(\Delta)}(\bftheta) = \exp(\kappa_{\bfZ^*(\Delta)}(\rmi \bftheta) )$  using a version of   Corollary \ref{wvagcor} 
    that is under $\PP$, specifically,
    \begin{align*}
        \kappa_{\bfZ^*(\Delta)}(\bftheta) =    \skal{(e^{\lambda\Delta}-1  )\bfeta}{\bftheta} + \kappa_{\bfV_0^*(\Delta)}(\bftheta) + \sum_{i=1}^n \kappa_{V_k^*(\Delta )}(\theta_k) ,\quad \bftheta \in \DDD_{\bfZ},
    \end{align*}
    where the law of $\bfV_0$, $V_k$, $k=1,\dots,n$, under $\PP$ are given in \eqref{wvagpropb}. 
    The numerical Fourier inversion method is described in \cite[Section 4]{MiSz17}. Using the notation there and in \cite[Section 5.1]{Lu21}, we set the number of grid points along each axis to be $N = 2^{12}$ for 1-dimensional FFT and $N=2^8$ for 2-dimensional FFT, and the spacing  $h_k$ is chosen such that $\bar x_k = 30 \sqrt{\myVar(Z^*_k(\Delta))}$, where $2\bar{x}_k$ is the width of the $\bfx$-grid along the $\bar{x}_k$ axis on which the pdf is outputted.  The factor $f_{\bfX(0)}(\bfx_0)$  has not been included in \eqref{likfn} as a single observation in the likelihood is negligible in practice for large sample sizes $m$, however to be more correct, \cite{Lu21} does include it and shows how it can be computed.

    To estimate $\bfvtheta$, we follow the 3-step procedure for the OU-WVAG process outlined in  \cite[Section 5.1]{Lu21}, which estimates $\lambda$ using method of moments on the lag 1 autocorrelation, the marginal parameters using MLE on the marginal components, and the joint parameters using MLE on the joint observations subject to a covariance constraint. Denote the estimated LDOUP parameters as $\wh\bfvtheta$.

    \begin{remark}
        In  \cite[Remark 9]{Lu21}, it is noted that there was no known closed-form formula for  $\Phi_{\bfZ^*(\Delta)}$ for a OU-WVAG process. We now have an analytical formula  and this substantially reduces the computational cost of MLE. While  MLE can be used to estimate all parameters jointly by \eqref{likfn} and the analytical formula, with the 3-step procedure meant to be an approximation, it turns out the latter still preforms better in practice as it allows a 1-dimensional FFT with large $N=2^{12}$ to estimate the marginal parameters accurately at a low computational cost, while the former method may be limited to smaller $N$, such as $N=2^8$, since 2-dimensional FFT has a higher computational cost. \halmos
    \end{remark}

    \subsubsection{Calibration}\label{calsec}

    Next, we estimate the MPR $\bfh$ by calibrating the model, which specifically solves
    \begin{align}
        \wh \bfh = \argmin_{\bfh\in\RR^2} \sum_{i=1}^{q}(O_i- E_i )^2,\label{optim}
    \end{align}
    where $O_i$ is the true price of the calibration instrument and $E_i$ is the predicted price. The calibration instruments are forward or call prices on $S_1$ and $S_2$ observed at the current time $T_e$.  For calibrating to the forward prices, the times to maturity, $\tau_i$, $i=1,\dots, 20$,  are equally spaced values from 0.125 to 2.5. For calibrating to the call prices, the time to maturity is $0.5$ and the strike prices, $K_i$, $i=1\,\dots,20$, are equally spaced values from 25 to 175. Note that $\bfh$ is estimated jointly over the calibration instruments on $S_1$ and $S_2$, so $q=40$.
    
    The forward and call prices are computed using \eqref{fwdpr} and \eqref{cm}, respectively, where the true price uses the true parameters $\bfb$, $\bfvtheta$ and the predicted price uses the estimated parameters  $\wh\bfb$, $\wh\bfvtheta$. To evaluate \eqref{cm}, we use numerical integration rather than  FFT as  done by \cite{CM99}. A discussion on the improved accuracy and speed of this is given in \cite[Section 15.2.3]{pasc}.
    
    We refer to estimating the model parameters $\bfb$, $\bfvtheta$ and calibrating $\bfh$ jointly as fitting the model.

    \begin{remark} \label{mcmmrem}
        Our approach estimates $\bfvtheta$ under $\PP$ and calibrates $\bfh$ under $\QQ_{\bfh}$, in contrast to the approach using the mean-correcting martingale measure, such as in Gardini and Sabino \cite{gs22}. The latter requires  the  model be specified under a risk-neutral measure $\QQ$ rather than $\PP$, and then calibrated to forward or call prices. An issue then arises with the estimation of joint parameters, such as those affecting the correlation, which should be calibrated to the prices of other multi-asset energy derivatives, but these may not be liquidly traded in the market.   To address this, one way is to calibrate the joint parameters such that the historical correlation of the log returns of $\bfS$ under $\PP$ matches that of the model correlation under $\QQ$ as a proxy. This effectively assumes that the $\PP$ and $\QQ$ correlations of the log return are equal, which is inconsistent with our model. More generally, one should expect that any change of measure can potentially change the correlation. An additional issue of this being an underdetermined problem, as there are 2 joint parameters but 1 correlation constraint, leaving 1 degree of freedom, is raised in \cite[Footnote 2]{gs22}. We avoid this method.

        In contrast, our method combines both the observed energy prices $\bfS$ under $\PP$ and univariate energy derivative prices to fit the model in a theoretically sound way. In fact, it is unclear whether our model can be estimated only by calibrating to derivative prices, as it may be the case that   $\bfvtheta$ and $\bfh$ would become unidentifiable. This issue arises if energy prices follow the simpler Black-Scholes model.  
        
        Interestingly, $h_2$ affects the distribution of $Z_1$ under $\QQ_{\bfh}$ as noted in Remark \ref{notwvagrem}, so it seems possible in principle to calibrate $\bfh$ to energy derivatives on $S_1$ only, which is not possible in the model of Gardini and Sabino.  \halmos
    \end{remark}

    \subsubsection{Pricing}\label{pricesec}

    Finally, we  discuss the FFT method for numerically evaluating the spread option pricing formula in \eqref{spreadint}, which can be applied using the true or fitted model.   While the implementation details are already given in \cite{HZ10}, and also \cite{AlSc2018}, here we provide additional explanation on how we obtain the spread option price for a panel of strike prices $K_i$, with a single FFT evaluation.
    
    Consider a 2-dimensional grid in the $\bftheta$-domain and $\bfx$-domain. This is specified by $N$, the number of grid points along each axis, and $\overline{\theta}_k$, $k=1,2$, where $2\overline{\theta}_k$ is the width of the grid along the $\theta_k$-axis. If $\overline{\theta}_1 =\overline{\theta}_2$, denote this value by $\overline{\theta}$. Then the $\bftheta$-grid and $\bfx$-grid are arrays whose entries are 
    \begin{align*}
        [\bftheta]_\bfi &=  \left( \delta_{\theta_1} \left(-\frac{N}{2} + i_1-1 \right),  \delta_{\theta_2} \left(-\frac{N}{2} + i_2-1 \right)\right), \quad \bfi = (i_1,i_2) \in \{1,\dots,N\}^2,\\
        [\bfx]_{\bfj} &=  \left( \delta_{x_1} \left(-\frac{N}{2} + j_1-1 \right),  \delta_{x_2} \left(-\frac{N}{2} + j_2 -1 \right)\right), \quad \bfj=(j_1,j_2) \in \{1,\dots,N\}^2,
    \end{align*}
    where the  $\theta_k$-spacing is $\delta_{\theta_k} = 2\overline{\theta}_k/N$ and the $x_k$-spacing is $\delta_{x_k} = 2\pi/(N\delta_{\theta_k})$ (note that $\overline{\theta}$, $\delta_{\theta}$, $\delta_{x}$ here are  $\overline{u}$, $\eta$, $\eta^*$ in \cite{HZ10}).

    At the current time $t=T_e$, for the log energy price $\log \bfS(t)=[\bfx]_{\bfj^*}$ on the $\bfx$-grid for some $\bfj^*$,  it is shown in \cite[Section 3]{HZ10} that the price of the spread option is approximated by
    \begin{align}
        f_{t}(\bfS(t),1) \approx \frac{ e^{-rT}} {(2\pi)^2}   \delta_{\theta_1}\delta_{\theta_2}N^2  (-1)^{j^*_1+j^*_2} e^{-\skal{\bfepsilon}{[\bfx]_{\bfj^*}}} [\mathscr{F}^{-1} \{\bfH\}]_{\bfj^*}, \label{fft}
    \end{align}
    where $\mathscr{F}^{-1}$ is the normalized 2-dimensional discrete inverse Fourier transform and $\bfH\in \CC^{N\times N}$ has entries
    \begin{align*}
        [\bfH]_\bfj = (-1)^{j_1+j_2} \exp( \kappa^{\QQ_{\bfh}}_{\bfR(t,T)}(- \bfepsilon + \rmi [\bftheta]_\bfj) )\hat{P}( [\bftheta]_\bfj + \rmi \bfepsilon ).
    \end{align*}
    The approximation improves by taking a large $N$ and $\overline{\theta}$.

    Now suppose we want to price the spread option for a general strike price $K>0$ using \eqref{sprid}, then we require $\log(\bfS(t)/K)=[\bfx]_\bfj$ hold for some $\bfj$, that is, it is on the $\bfx$-grid. If  $\log(\bfS(t)/K)$ is bounded by the $\bfx$-grid but not on it, taking the nearest grid point can lead to sizable pricing errors, so other approaches are considered.
    
    \begin{itemize}
        \item[1.] Using bicubic interpolation on the $\bfx$-grid to compute to $f_{t}(\bfS(t)/K,1)$.
        \item[2.] Adjust the  $\bfx$-grid so that $\log(\bfS(t)/K)$ is on it. Do this by setting $\bfj^*$ as  the index of the grid point nearest $\log(\bfS(t)/K)$, and then adjust the values of $\delta_{\theta_1},\delta_{\theta_2}$ such that $\log(\bfS(t)/K) = [\bfx]_{\bfj^*}$ holds, provided all elements of $[\bfx]_{\bfj^*}$ are nonzero.
    \end{itemize}
    
    To compute the price of a spread option $f_t( \bfS(t), K_i)$ for a panel of strike prices  $K_i$, $i=1,\dots,q$, on a single FFT evaluation, we do both. For a particular strike price of $K$, we adjust the $\bfx$-grid  so that $\log(\bfS(t)/K)$ is on it. Then a single FFT evaluation outputs  $(f_t( [\bfx]_\bfj, 1))\in\RR^{N\times N}$, from which we obtain $f_t(\bfS(t)/K_i,1)$ using bicubic interpolation. Finally,  $f_{t}(\bfS(t),K_i) = K_i f_t(\bfS(t)/K_i,\allowbreak 1)$ by \eqref{sprid}. The grid is always adjusted so that $K = 3.6$ is on it.

    \begin{remark}
        Note that the FFT used for spread option pricing is different from the one used in the Fourier inversion method for estimation. The former uses the discrete inverse Fourier transform, while the latter uses the discrete Fourier transform. \halmos
    \end{remark}

    Lastly, pricing spread options using the Monte Carlo method requires  simulating $\bfX(T_o)$ conditional on $\bfX(T_e)=\log \bfS(T_e)-\bfLambda(T_e)$. This is  done by either discretizing the SDE \eqref{ldoup} or discretizing the integral $\bfZ^*(\tau)$ in \eqref{logrettau}, where we use $\wt n = 10^5$ subintervals  for both.

        \subsection{Real data example} \label{realimpsec}

        In the real data example, we fit a model to Australian electricity and base load futures prices and then price spread options under the model. Our implementation mostly follows Section \ref{simimpsec}, except for the differences outlined here. All unmentioned details remain the same as before.

        The data consists of  8 years of  daily average electricity prices for the Australian states NSW  and Victoria in the NEM sourced from \href{https://www.aemo.com.au/energy-systems/electricity/national-electricity-market-nem/data-nem/data-dashboard-nem}{AEMO} (Australian Energy Market Operator). These prices were obtained from 2018-01-05 to 2026-01-05. This is used to estimate the model parameters. The electricity price data consists of 2923 daily observations, so that $m= 2922$ and $\Delta = \frac{1}{365.25}$ accounting for leap years.

        Let $\bfS' = (S'_1,S'_2)$, where $S'_1(t)$  and  $S'_2(t)$ are the NSW  and Victoria daily average electricity at time $t$, respectively. The observations of $S_1'$ are always positive, but 5.9\% of the observations of $S_2'$ are negative with a minimum price $-\$73.02$, so a geometric model cannot be used to model $S_2'$ directly.   Using  Remark \ref{parity} to deal with this, we let  $S_1(t) = S'_2(t)+c$, $S_2(t)=S_1'(t)$, $t\geq0$, and fit the geometric model to $\bfS =(S_1,S_2)$, where $c$ is a constant that makes the observations of $S_1$ positive. The NEM has a negative price floor, so such a $c$ exists. We take $c=100$, however, different choices lead to different models.

        We assume the seasonality function is
        \begin{align*}
            \Lambda_k(t) = {}&b_{0k} + b_{1k}t + b_{2k}\cos\left(\frac{2\pi}{p_1 }t\right)+b_{3k}\sin\left(\frac{2\pi}{p_1}t\right)\\
            \quad& + b_{4k} \cos\left(\frac{2\pi}{p_2 }t\right)+b_{5k}\sin\left(\frac{2\pi}{p_2 }t\right), \quad t\in[0,T],
        \end{align*}
        with periods $p_1=1$ and $p_2=7\Delta$, which correspond to yearly and weekly cyclical components, respectively. The inclusion of weekly seasonality in electricity prices makes the model more realistic.      
        
        The risk-free rate is assumed to be $r=0.03715$. This only affects the discount term for the pricing of spread options and  plays no role in fitting the model.
        
        At the current time $T_e=8$, we price spread options on the original prices $(S_1',S_2')$ whose payoff is given by \eqref{pp}, where the strike prices are $K'_i$, $i=1,\dots, 13$, which are equally spaced values from 20 to 80, and the time to maturity is $\tau = 0.5$. This is done by computing the price of spread options on $(S_1,S_2)$ using \eqref{parrel}.

        The data also consists of base load calendar quarter electricity futures prices for NSW and Victoria with 8 contract periods Q1 2026 to Q4 2027 on the day 2026-01-05 sourced from \href{https://www.asxenergy.com.au/}{ASX Energy} (Australian Stock Exchange). This is used to calibrate the model, so there are $q=16$ calibration instruments. The base load futures prices on $S_k$, which is $\overline{F}_k(t,T_1,T_2)$, is computed using \eqref{futadj} and \eqref{futuresprice} for $k=1,2$. For example, since the current time is 2026-01-05, the Q1 2026 base load futures price has $t=T_e+5\Delta$, $T_1 = T_e$, $T_2=T_e+90\Delta $, and is computed using 5 observed electricity prices from 2026-01-01 to 2026-01-05 by  \eqref{futurespricessum}.
        
        In the FFT method, we set $N=2^{8}$, $\bar{\theta}=40$, and the grid is always adjusted so that $K = 50$ is on it.

    \section{Simulated data results}\label{ressec}

    We now present results from simulated data examples based on the implementation details in Section \ref{simimpsec}. Results in Sections \ref{trueparsec1}--\ref{trueparsec2} are under the true model. Results in Sections \ref{simparsec1}--\ref{priceres}  are from a simulation study using fitted models.

    \subsection{Spread option pricing under true parameters}\label{trueparsec1}

    Using the true parameters,  Table \ref{tab1} gives the spread option price using the FFT method with $N=2^{11}$ grid points along each axis and $\overline{\theta}=80$, and the Monte Carlo method with $N=10^6$ simulations and $\wt n = 10^5$ subintervals by discretizing the SDE.
    
    \begin{table}[htb]
        \centering
        \begin{tabular}{cccc}
            \hline
            $K_i$ & FFT & Monte Carlo & 95\% confidence interval    \\
            \hline
            0.4 & 9.31079 & 9.32370 & $( 9.28429 ,9.36311)$\\
            1.2 & 9.03505& 9.04812& $(9.00917, 9.08708)$\\
            2.0 & 8.76650 &8.77975& $(8.74125, 8.81826)$\\
            2.8 &  8.50503& 8.51849& $(8.48044, 8.55655)$\\
            3.6 &  8.25053& 8.26418& $(8.22657, 8.30179)$\\
            4.4 &   8.00289& 8.01663& $(7.97947, 8.05379)$\\
            5.2 & 7.76200& 7.77583 &$(7.73911, 7.81255)$\\
            6.0 &  7.52773& 7.54163 &$(7.50535, 7.57790)$\\
            6.8 &   7.29998& 7.31393&$( 7.27808 ,7.34977)$\\
            7.6 &      7.07861&7.09258& $(7.05717 ,7.12799)$\\
            8.4 & 6.86350& 6.87758 &$(6.84261 ,6.91256)$\\
            9.2 & 6.65454 &6.66877&$  (6.63422, 6.70332)$\\
            10.0 &   6.45158& 6.46594&$  (6.43181, 6.50006)$\\
            
            \hline
        \end{tabular}
        \caption{Spread option prices over various strike prices $K_i$ using the FFT method with $N=2^{11}$, $\overline{\theta}=80$ and the Monte Carlo method with $N=10^6$ and $\wt n=10^5$.}\label{tab1}
    \end{table}
    
    This shows the spread option price using the FFT method closely agrees with the Monte Carlo method over the panel of strike prices. The 95\% confidence intervals are also given, which is narrow due to the large number of simulations, with all prices from the FFT method within it.

    For the Monte Carlo method, while simulating the integral should be more accurate than the SDE when it can be done exactly (for example, the model in \cite[Section 3.1]{Lu21}), the choice is less clear when the integral also needs to be discretized. In this example, for small $\wt n=10^2$, the average error is 0.67664 for discretizing the integral and  0.3700 for the SDE, both upward biased with the true price  outside the 95\% confidence intervals in all cases. But for large $\wt n=10^5$, there is no significant difference. In our simulation, the average absolute difference between the prices of the two discretization methods is $3.1\times 10^{-4}$, well within any difference detectable by the 95\% confidence intervals in Table \ref{tab1}. Thus, we conclude that discretizing the SDE is more accurate and less biased than discretizing the integral, particularly for small  $\wt n$, while the difference is negligible for large $\wt n$.
    
    Next,   the average error  over the panel of strike prices is defined as
    \begin{align}
        \text{Average error}= \frac{1}{q}\sum_{i=1}^q |O_i-E_i|,\label{avgerr}
    \end{align}
    where $O_i$ is the true price of the spread option with strike price $K_i$, and $E_i$ is  the price of the spread option computed in another way. We take the true price to be from the FFT method with $N=2^{11}$, $\bar{\theta}=80$, and it is compared with the prices using  the discretization parameters $N=2^k$, $k=6,7,\dots 10$, $\overline{\theta}=40,80$. Figure \ref{fig2} shows how this affects the average error of the FFT method.
    
    \begin{figure}[htb]
        \begin{center}
            
            \includegraphics{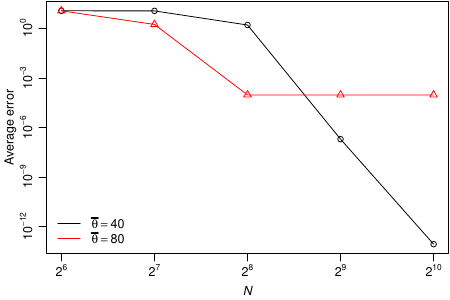}
            \caption{Average error using the FFT method for $N=2^k$, $k=6,7,\dots 10$, and $\overline{\theta}=40,80$. The plot is on a log scale.}\label{fig2}
        \end{center} 
    \end{figure}
    
    The points in Figure \ref{fig2} with large error only occur for small $N$, and in fact include situations for small $K_i$, where the spread option price is negative, or where no appropriate choice of $\bfepsilon$ in the sense of Remark \ref{epsrem} can be readily found. But as $N$ and  $\overline{\theta}$ increase, the average error vanishes, as expected. We see that  $N=2^8$, $\overline{\theta}=40$  is enough to produce highly accurate results with average error in the order of $10^{-5}$. Furthermore, the way the error reduces and flattens out is consistent with \cite[Figures 1--3]{HZ10}.

    \subsection{Spread option  parameter sensitivities}\label{trueparsec2}
    
    \begin{figure}[!tb]
        \begin{center}
            \includegraphics{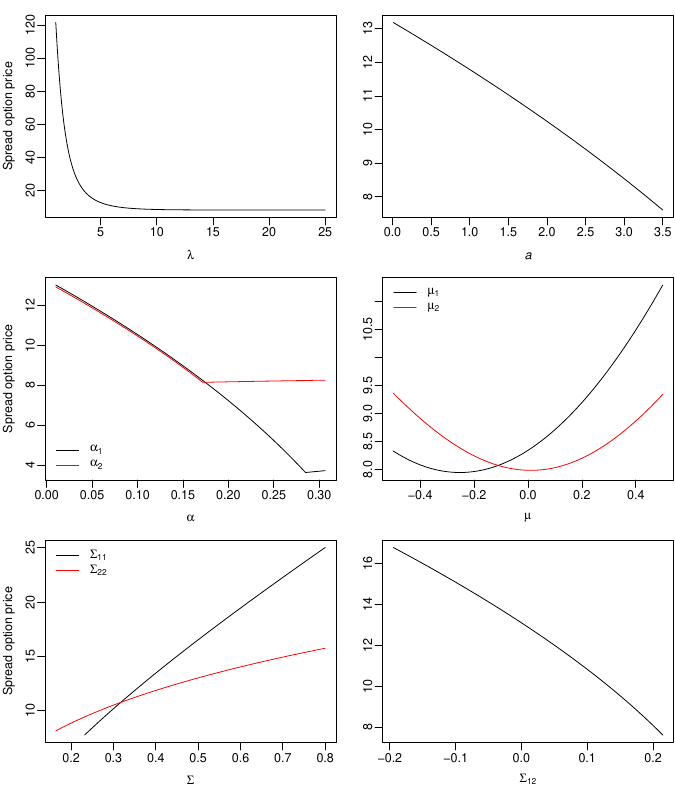}
            \caption{Spread option price with $K=3.6$ as a function of the LDOUP parameters.}\label{fig3}
        \end{center}
    \end{figure}
    
    Figure \ref{fig3}  shows  the effect of the LDOUP parameters on the spread option price with strike price $K=3.6$ and all other parameters held constant. We refer to \cite[Remark 4]{MiSz17} for a discussion about the interpretation of the parameters of the WVAG process. The parameters $a,\alpha_1,\alpha_2, \Sigma_{12}$ are plotted over their full range subject to the parameter constraints, the same is also true for the lower limit of $\Sigma_{11},\Sigma_{22}$.

    The spread option price decreases as the rate of mean reversion $\lambda$ increases, as this reduces the variance of the spread.  The plots for  the kurtosis parameters $\alpha_1$ and $\alpha_2$ have a point of nondifferentiability at $\alpha_2=\alpha_1$ due to the minimum function occurring in various terms of 
    \begin{align*}
        \myCov_{\QQ_{\bfh}}(X_1(T_o),X_2(T_o)\given \FFF_{T_e}) = \frac 12 { (1- e^{-2\lambda \tau})}\left({ [\bfSigma_{\bfh}]}_{12} + \frac{{[\bfmu_{\bfh}]}_{1}{[\bfmu_{\bfh}]}_{2}}{a} \right), 
    \end{align*}
    which follows from Corollary \ref{vgn}, and where for example, ${[\bfSigma_{\bfh}]}_{12} = a \Sigma_{12}(\alpha_1\wedge \alpha_2)/K_{\bfh}$. This covariance in turn affects the variance of the spread  under $\QQ_{\bfh}$.
    Some parameters may have  surprising effects given that they affect moments other than  what is considered their primary one.  For example, it may be expected that increasing the skewness parameters $\mu_1$ and $\mu_2$ have opposite effects, but actually both tend to increase the variance of the spread under $\QQ_{\bfh}$ in asymmetric ways, which causes the price to increase. As variance parameter $\Sigma_{22}$ increases, this increases the variance of  the spread because $\Sigma_{11}, \Sigma_{12}$ are held constant, but the correlation ${\Sigma_{12}}/\sqrt{\Sigma_{11}\Sigma_{22}}$ is not, and hence the price increases.

    \subsection{Fitting the model}\label{simparsec1}

    We consider a simulation study based on 1000 simulated sample paths. Recalling the details from Section \ref{imple}, each sample path has 2001 observations for estimating the model parameters, and 40 forward prices or call prices for calibration.

    For 1 of the 1000 simulations, Figures \ref{fig4} and \ref{fig5}  plot the log price sample path $Y_k(t)$, and compare the true and estimated seasonality function $\Lambda_k(t)$, as well as the true and fitted conditional distribution of the log returns $R_k(T_e,T_o)\given \FFF_{T_e}$ under $\QQ_\bfh$ for $k=1,2$. Define the risk-neutral drift as $\wt \bfmu=(\wt \mu_1,\wt \mu_2)$, where
    \begin{align}
        \wt \mu_k:= {}&  \EE_{\QQ_{\bfh}}[R_k(T_e,T_o)\given \FFF_{T_e}] \label{rndr}\\
        ={}&  \Lambda_k(T_o) + X_k(T_e)e^{-\lambda \tau }-Y_k(T_e)+(1-e^{-\lambda \tau})(\eta_k + {[\bfmu_\bfh]}_k  + \mu_{h_k} ),\nonumber
    \end{align}
    $k=1,2$, due to \eqref{logrettau} and \eqref{match0}, and this is shown in  Figure \ref{fig5}.

    \begin{figure}[htb]
        \begin{center}
            \includegraphics{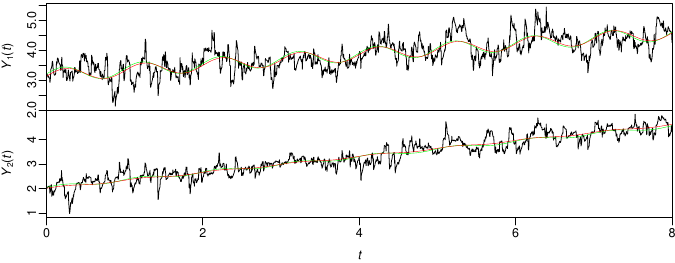}
            \caption{Sample path of the log price $Y_k(t)$ (black), true seasonality function $\Lambda_k(t)$ (green) and estimated seasonality function $\wh\Lambda_k(t)$ (red) for $k=1,2$.}\label{fig4}
        \end{center} 
    \end{figure}
    
    \begin{figure}[htb]
        \begin{center}
            \includegraphics{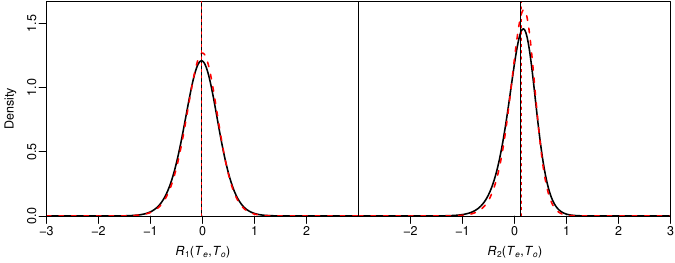}
            \caption{ True (black) and estimated (red) pdf of $R_k(T_e,T_o)\given \FFF_{T_e}$ under $\QQ_\bfh$ and the corresponding means  $\wt \mu_k$ (vertical line) for $k=1,2$.}\label{fig5}
        \end{center}
    \end{figure}

    In both Figures \ref{fig4} and \ref{fig5}, the true and estimated quantities are close, although the fitted distribution of $R_2(T_e,T_o)$ appears to have slightly lower variance than the true distribution, while the risk-neutral drifts are virtually identical.

    Table \ref{esttable} gives the true value, mean estimate, estimated standard error, and estimated relative root mean square error (RRMSE) of the seasonal parameters $\bfb$, LDOUP parameters $\bfvtheta$,  MPRs  $\bfh^f = (h_1^f,h_2^f)$  and $\bfh^c =(h_1^c,h_2^c)$  from calibrating to forward and call prices, respectively, and their corresponding risk-neutral drifts $\wt \bfmu^f=(\wt \mu_1^f,\wt \mu_2^f)$, $\wt \bfmu^c=(\wt \mu_1^c,\wt \mu_2^c)$.

    \begin{table}[htb]
        \centering
        \begin{tabular}{
                c
                S[table-format=2.5]  
                S[table-format=2.5]  
                S[table-format=1.5]  
                S[table-format=4.5]  
            }
            \hline
            Parameter & {True value} & {Mean estimate} & {Standard error} & {RRMSE (\%)} \\
            \hline
            $b_{01}$     & 3.16132 & 3.16012 & 0.08181 & 2.58672 \\
            $b_{11}$     & 0.17500 & 0.17512 & 0.01737 & 9.91921 \\
            $b_{21}$     & 0.04385 & 0.04453 & 0.05355 & 122.08732 \\
            $b_{31}$     & 0.22986 & 0.22998 & 0.05593 & 24.32271 \\
            $b_{02}$     & 2.04056 & 2.03718 & 0.07231 & 3.54595 \\
            $b_{12}$     & 0.31390 & 0.31433 & 0.01534 & 4.88649 \\
            $b_{22}$     & 0.01257 & 0.01236 & 0.04685 & 372.67504 \\
            $b_{32}$     & 0.06587 & 0.06285 & 0.04746 & 72.15719 \\
            \hline
            $\lambda$    & 18.25000 & 19.72150 & 1.88547 & 13.10125 \\
            $a$          & 3.15854 & 3.78378 & 0.48648 & 25.07668 \\
            $\alpha_1$   & 0.17183 & 0.16369 & 0.01692 & 10.91866 \\
            $\alpha_2$   & 0.28494 & 0.23210 & 0.02699 & 20.82233 \\
            $\mu_1$      & -0.03071 & -0.03114 & 0.01511 & 49.19107 \\
            $\mu_2$      & -0.20335 & -0.18828 & 0.03601 & 19.19018 \\
            $\Sigma_{11}$& 0.24598 & 0.22352 & 0.02720 & 14.33557 \\
            $\Sigma_{22}$& 0.17045 & 0.14425 & 0.02185 & 20.00932 \\
            $\Sigma_{12}$& 0.19452 & 0.16270 & 0.02204 & 19.89720 \\
            \hline
            $h_{1}^f$    & -0.10000 & -0.11860 &0.60465 & 604.63349 \\
            $h_{2}^f$    & -0.03000 & 0.01899 & 0.80720 & 2694.26668 \\
            $h_{1}^c$    & -0.10000 & -0.09548 & 0.58939 & 589.11447 \\
            $h_{2}^c$    & -0.03000 & 0.05402 & 0.77961 & 2612.46743 \\
            \hline
            $\widetilde{\mu}_1^f$ & -0.02806 &  -0.02838 & 0.05864 & 208.87460 \\
            $\widetilde{\mu}_2^f$ & 0.11560 & 0.11901 & 0.04982 &  43.17967 \\
            $\widetilde{\mu}_1^c$ & -0.02806 & -0.02021 & 0.01131 & 49.05332 \\
            $\widetilde{\mu}_2^c$ & 0.11560 & 0.12612 & 0.00985 & 12.46426 \\
            \hline
        \end{tabular} 
        \caption{True value, mean estimate, standard error, and RRMSE for the seasonal parameters, LDOUP parameters, MPRs, risk-neutral drifts estimated over 1000 simulation with 2001 observations and 40 calibration instruments per simulation.}\label{esttable}
    \end{table}

    From these results, most model parameters are estimated with moderately low RRMSE of approximately 10--30\%, other than those with true values near 0. Related to this, note that errors in  estimating the model parameters  are compensated for by the  calibration. For instance, $h_1^c$ and $h_2^c$ have RRMSEs  of  589\% and 2612\%, which can take them far from their true value, so that $\wt \mu^c_1$ and  $\wt \mu^c_2$  can have small RRMSEs of 49\% and 12\%, close to their true value. In fact, on an absolute scale as seen by the vertical line in Figure \ref{fig5}, the differences between the true and estimated values of $\wt \mu^c_1$ and $\wt \mu^c_2$ are negligible. Ultimately, it is $\wt \bfmu$ that has a first-order effect on the energy derivative prices, and the fitted models should primarily  be judged on how accurately they price spread options.

    \begin{figure}[htb]
        \begin{center}
            \includegraphics{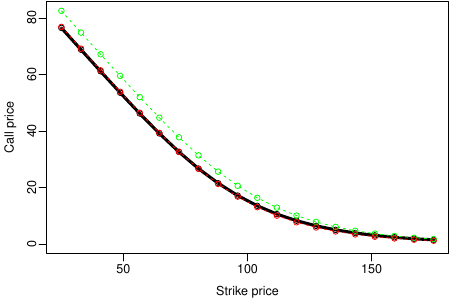}
            \caption{The true call prices on $S_1$  (black) over various strike prices $K_i$, and the predicted call prices using the true MPR $\bfh$ (green) and  calibrated MPR $\wh \bfh^c$ (red).}\label{figlast}
        \end{center}
    \end{figure}
    
    Figure \ref{figlast} shows the results of  calibrating to call prices. Based on the estimated model parameters, the calibrated MPR $\wh\bfh^c$  gives predicted call prices with a very close fit to the true call prices on $S_1$, but the true MPR  $\bfh$ does not.   While Figures  \ref{fig5} and \ref{figlast}  are based on a single simulation, these are typical results seen in most simulations.       Only results for calibrating to call prices on $S_1$ are plotted, although we calibrated to call prices  on $S_1$ and $S_2$. A similarly close match holds for call prices on $S_2$, while calibration to forwards does not produce such a close match.

    \subsection{Spread option pricing under the fitted model}\label{priceres}

    Finally, over the 1000 simulations, the results from using the fitted model to price spread options are presented in Table \ref{spreadresults}, with  the calibration to forward and call prices shown separately. The spread option prices are computed by  the FFT method with  $N=2^{8}$, $\bar{u}=40$,  which gives sufficient accuracy as noted in Section \ref{trueparsec1}. The true prices are from Table \ref{tab1}, computed with $N=2^{11}$, $\bar{u}=80$. The mean prices  and their  RRMSEs are given.

    \begin{table}[!htb]
        \centering
        \centerline{
            \begin{tabular}{cccccc}
                \hline
                & & \multicolumn{2}{c}{Calibrated to forwards} & \multicolumn{2}{c}{Calibrated to calls} \\ 
                $K_i$ & True price & Mean fitted price & RRMSE (\%) & Mean fitted price & RRMSE (\%) \\
                \hline
                0.4  & 9.31079 & 8.28003 & 29.28306 & 8.14641 & 16.05140 \\
                1.2  & 9.03505 & 8.01567 & 29.60813 & 7.88303 & 16.37409 \\
                2.0  & 8.76650 & 7.75871 & 29.93505 & 7.62714 & 16.69879 \\
                2.8  & 8.50503 & 7.50936 & 30.25915 & 7.37888 & 17.02162 \\
                3.6  & 8.25053 & 7.26751 & 30.58043 & 7.13819 & 17.34162 \\
                4.4  & 8.00289 & 7.03275 & 30.90307 & 6.90463 & 17.66241 \\
                5.2  & 7.76200 & 6.80489 & 31.22491 & 6.67802 & 17.98405 \\
                6.0  & 7.52773 & 6.58422 & 31.54294 & 6.45865 & 18.30082 \\
                6.8  & 7.29998 & 6.37004 & 31.86120 & 6.24581 & 18.61972 \\
                7.6  & 7.07861 & 6.16257 & 32.17869 & 6.03969 & 18.93635 \\
                8.4  & 6.86350 & 5.96113 & 32.49764 & 5.83965 & 19.25737 \\
                9.2  & 6.65454 & 5.76624 & 32.81268 & 5.64619 & 19.57346 \\
                10.0 & 6.45158 & 5.57759 & 33.12549 & 5.45901 & 19.88679 \\
                \hline
        \end{tabular}}
        \caption{Spread option prices for the fitted models and their RRMSE over various strike prices $K_i$ computed by the  FFT method with  $N=2^{8}$, $\bar{\theta}=40$. The true prices are from Table \ref{tab1}.} \label{spreadresults}
    \end{table}
    
    The results show that the spread options  are priced less accurately when calibrated to forwards, with an average RRMSE of 31.2\%, than to calls, where it is 18.0\%. This shows how errors in fitting the model, from 2001 observations and 40 calibration instruments, propagate to pricing errors.  As the sample size increases, so that the model parameters are estimated more accurately, the pricing error should approach 0.

    It should be unsurprising that calibrating to calls leads to lower pricing error since the forwards calibrate over various maturities, whereas the calls calibrate to the same maturity and  different strikes, which implies, at least in part, the marginal component distribution and joint distribution of $\bfS(T_o)\given \FFF_{T_e}$. Both calibration methods produce prices that are downward biased. Calibrating to calls  has slightly larger bias and substantially lower standard error here.

       \section{Real data results}\label{realdata}
        
        We present results from a real data example based on the implementation details in Section \ref{realimpsec}, where we fit the model to Australian electricity and base load futures prices, and price spread options.

        Figure \ref{fig7} plots the sample path of the log price  $Y_k(t)$ and the estimated seasonality function $\wh\Lambda_k(t)$, in which a small weekly cyclical component is visible. Table \ref{realesttable} gives the estimated parameters of the fitted model, where $(\wt\mu_1,\wt\mu_2)$ are the risk-neutral drifts defined in \eqref{rndr}.

        \begin{figure}[htb]
            \begin{center}
                \includegraphics{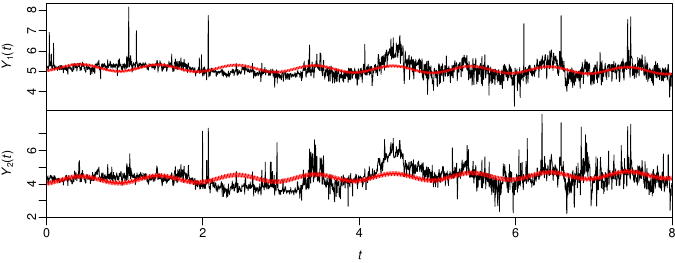}
                \caption{Sample path of the log price $Y_k(t)$ (black) and the estimated seasonality function $\wh\Lambda_k(t)$ (red) for $k=1,2$.}\label{fig7}
            \end{center} 
        \end{figure}

        \begin{table}[!htb]
            \centering
            \begin{tabular}{ c @{\hspace{2em}} c }
                \begin{tabular}[t]{c S[table-format=2.5]}
                    \hline
                    Parameter & {Estimate} \\
                    \hline
                    $b_{01}$ &  5.19036 \\
                    $b_{11}$ & -0.01937 \\
                    $b_{21}$ & -0.14974 \\
                    $b_{31}$ &  0.07126 \\
                    $b_{41}$ & -0.01048 \\
                    $b_{51}$ & -0.06942 \\
                    $b_{02}$ &  4.22798 \\
                    $b_{12}$ &  0.04056 \\
                    $b_{22}$ & -0.19417 \\
                    $b_{32}$ &  0.08653 \\
                    $b_{42}$ & -0.01921 \\
                    $b_{52}$ & -0.13987 \\
                    \hline
                \end{tabular}
                
                &
                \begin{tabular}[t]{
                        c S[table-format=3.5]
                    }
                    \hline
                    Parameter & {Estimate} \\
                    \hline
                    $\lambda$     & 123.14798 \\
                    $a$           &   2.15520 \\
                    $\alpha_1$    &   0.34922 \\
                    $\alpha_2$    &   0.31290 \\
                    $\mu_1$       &  -0.01398 \\
                    $\mu_2$       &   0.13247 \\
                    $\Sigma_{11}$ &   0.20842 \\
                    $\Sigma_{22}$ &   0.49101 \\
                    $\Sigma_{12}$ &   0.26145 \\
                    \hline
                    $h_{1}$       &  0.80684 \\
                    $h_{2}$       & -0.48216 \\
                    \hline
                    $\widetilde{\mu}_1$ & 0.44854 \\
                    $\widetilde{\mu}_2$ & 0.78568 \\
                    \hline
                \end{tabular}
            \end{tabular}
            \caption{The estimate of the seasonal parameters, LDOUP parameters, MPRs, risk-neutral drifts.}\label{realesttable}
        \end{table}
        
        \begin{figure}[!htb]
            \begin{center}
                \includegraphics{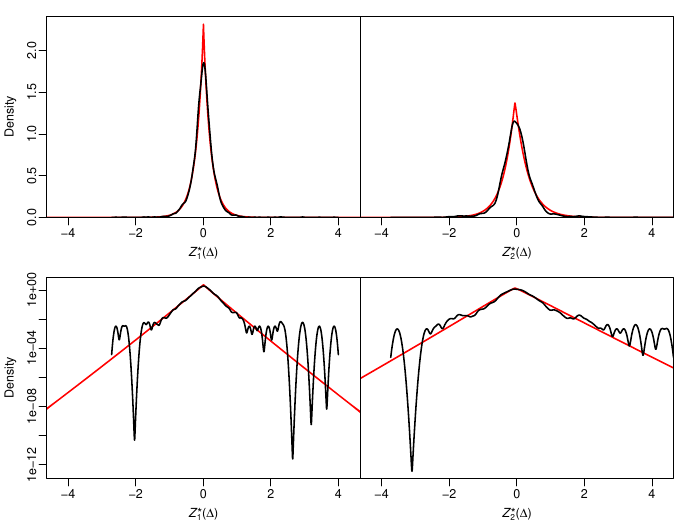}
                
                \caption{Kernel density estimate (black) and fitted (red) pdf of  $Z_1^*(\Delta)$ (left) and $Z_2^*(\Delta)$ (right) under $\PP$ on the original scale (top) and log scale (bottom).}\label{fig8}
            \end{center}
        \end{figure}

        To assess the goodness-of-fit, we plot in  Figure \ref{fig8} the pdf of the marginal components of the innovation term $\bfZ^*(\Delta)$ under $\PP$ obtained by  the kernel density estimate and the fitted distribution.  The observations of $\bfZ^*(\Delta)$ were obtained by \eqref{filteredobs} and \eqref{sim} with the estimate of $\lambda$. The plot effectively compares the empirical and fitted distribution of the residual for each marginal component on the original scale and log scale. The fit  appears to be reasonably good. As the MLE is not a robust estimation method, it has fit the tails well on the log scale, but at the expense of the fit in the center as shown on the original scale. In fact, if not using the MLE, it is possible to obtain a visually near-perfect fit on the original scale but a worse fit to the tails on the log scale. To assess the joint distribution, the theoretical and sample value of the correlation $\myCorr(\bfZ_1^*(\Delta),\bfZ_2^*(\Delta))$ are 0.55 and 0.47, respectively, which is reasonably close.

        Another way to assess goodness-of-fit is to compare the predicted and observed prices of the base load futures prices for each of the 8 contract periods on the day 2026-01-05 as shown in Figure \ref{fig9}. Note that from the base load futures prices on $S_1$, subtracting $c$ gives the base load futures prices for Victoria. The  average percentage error, which is defined by \eqref{avgerr} with summand $|O_i-E_i|/E_i$ instead, is 4.62\%. In addition, the predicted and observed futures curves have a qualitatively similar shape.  Overall, the results suggest the model is a reasonable fit to the data despite some imperfections.
        
        \begin{figure}[!htb]
            \begin{center}
                \includegraphics{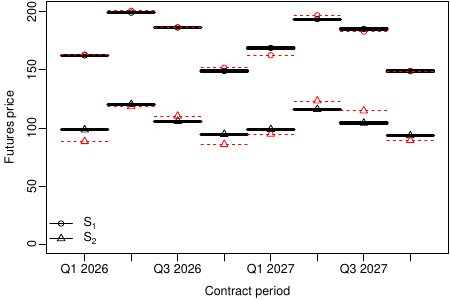}
                \caption{The observed base load futures price  (black) on $S_1$ and $S_2$ for the contract periods Q1 2026 to Q4 2027, and the predicted futures prices (red) from the fitted model.}\label{fig9}
            \end{center}
        \end{figure}

        The results here can also be compared to those in the literature. In \cite{bsb04}, a univariate LDOUP with a NIG BLDP was fit separately to electricity and gas price data, and their Figure 5 and 6 is similar to Figure \ref{fig8} here. They make the same finding regarding the fit to the center and tail. In  \cite{cf05}, a jump-diffusion BDLP is used instead, and their Figure 9 has a slightly worse fit compared to   Figure \ref{fig9} here.  In contrast to both these univariate models,  we have fit a multivariate model.

        Finally, using the fitted model on the day 2026-01-05, we give the price of the spread option on  $(S_1,S_2)$ with strike prices $K_i$, and the corresponding price of the spread option on the original prices $(S_1',S_2')$ with strike prices $K'_i$ in Table \ref{realspread}.

        \begin{table}[!htb]
            \centering
            \begin{tabular}{cccc}
                \hline
                $K_i$ & Price on $(S_1,S_2)$  & $K'_i$ &  Price on $(S'_1,S'_2)$ \\
                \hline
                20 & 66.32141 & 80 & 6.27083 \\
                25 & 62.08124 & 75 & 6.93864 \\
                30 & 57.93256 & 70 & 7.69793 \\
                35 & 53.88806 & 65 & 8.56142 \\
                40 & 49.96123 & 60 & 9.54256 \\
                45 & 46.16612 & 55 & 10.65543 \\
                50 & 42.51608 & 50 & 11.91336 \\
                55 & 39.02593 & 45 & 13.33119 \\
                60 & 35.70724 & 40 & 14.92047 \\
                65 & 32.57084 & 35 & 16.69204 \\
                70 & 29.62514 & 30 & 18.65432 \\
                75 & 26.87572 & 25 & 20.81287 \\
                80 & 24.32528 & 20 & 23.17042 \\
                \hline
            \end{tabular}
            \caption{Spread option prices  on  $(S_1,S_2)$ for the fitted models over various strike prices $K_i$, and the  corresponding spread option prices on $(S'_1,S'_2)$ over the  strike prices $K_i'$.} \label{realspread}
        \end{table}

    \section{Concluding remarks}\label{conc}

    The framework for pricing energy derivatives presented here has wider applicability beyond spread options. Following the explanations in Eberlein, Glau and Papapantoleon \cite{EbGlPa2010} and \cite{KLW11,AlSc2018}, the FFT method can similarly be applied to other payoff functions with a known analytical Fourier transform. Moreover, any European-style payoff in our model can readily be priced using the Monte Carlo method. Other extensions, such as generalizing the speed of mean reversion $\lambda$ to a matrix, using a generalized hyperbolic (GH) process, or other processes as the BDLP, would likely lead to a better fit in Figure \ref{fig8} because the GH process has 1 additional shape parameter and includes the VG process as a subclass. Another direction is to consider an arithmetic model.  However, these are beyond the scope of this paper.

    \subsection*{Acknowledgments}
    Kevin Lu thanks Mesias Alfeus for discussions  relating to this work. We thank the referees for their comments.

    \bibliographystyle{plain}  
    \bibliography{bibliography}  
    
\end{document}